\documentclass{article}

\def\noheaderplainsetup{

\topmargin=0pt \headheight=0pt \headsep=0pt  \oddsidemargin=0pt \evensidemargin=0pt  \textheight=8.9truein \textwidth=6.2truein}   

\noheaderplainsetup

\usepackage{amsfonts}

\begin{document}

%     MISC.:

\newcommand{\atime}{\mathbb{AT}}
\newcommand{\xx}{\wp}
\newcommand{\col}[1]{\mbox{$#1$:}}

\newcommand{\seq}[1]{\langle #1 \rangle}           % sequence: <...>

%     OPERATORS:

\newcommand{\mla}{\mbox{{\Large $\wedge$}}}
\newcommand{\mle}{\mbox{{\Large $\vee$}}}

\newcommand{\pst}{\mbox{\raisebox{-0.01cm}{\scriptsize $\wedge$}\hspace{-4pt}\raisebox{0.16cm}{\tiny $\mid$}\hspace{2pt}}}
\newcommand{\gneg}{\neg}                  %game negation
\newcommand{\mli}{\rightarrow}                     %strong reduction
\newcommand{\cla}{\mbox{\large $\forall$}}      %blind universal quantifier
\newcommand{\cle}{\mbox{\large $\exists$}}        %blind existential quantifier
\newcommand{\mld}{\vee}    %multiplicative disjunction
\newcommand{\mlc}{\wedge}  %multiplicative conjunction
\newcommand{\ade}{\mbox{\Large $\sqcup$}}      %additive existential quantifier
\newcommand{\ada}{\mbox{\Large $\sqcap$}}      %additive universal quantifier
\newcommand{\add}{\mbox{\large $\sqcup$}}                      %additive disjunction
\newcommand{\adc}{\mbox{\large $\sqcap$}}                      %additive conjunction

\newcommand{\tlg}{\bot}               %classical \bot; trivially lost elementary game
\newcommand{\oo}{\bot}               %classical \bot; trivially lost elementary game
\newcommand{\twg}{\top}               %classical \top; trivially won elementary game
\newcommand{\pp}{\top}               %classical \top; trivially won elementary game
\newcommand{\st}{\mbox{\raisebox{-0.05cm}{$\circ$}\hspace{-0.13cm}\raisebox{0.16cm}{\tiny $\mid$}\hspace{2pt}}}
\newcommand{\bst}{\st^{\aleph_0}}
\newcommand{\cost}{\mbox{\raisebox{0.12cm}{$\circ$}\hspace{-0.13cm}\raisebox{0.02cm}{\tiny $\mid$}\hspace{2pt}}}
\newcommand{\bcost}{\cost^{\aleph_0}}
\newcommand{\pcost}{\mbox{\raisebox{0.12cm}{\scriptsize $\vee$}\hspace{-4pt}\raisebox{0.02cm}{\tiny $\mid$}\hspace{2pt}}}

%   NUMERATED ITEMS and ENVIRONMENTS

\newtheorem{theoremm}{Theorem}[section]
\newtheorem{thesiss}[theoremm]{Thesis}
\newtheorem{definitionn}[theoremm]{Definition}
\newtheorem{lemmaa}[theoremm]{Lemma}
\newtheorem{propositionn}[theoremm]{Proposition}
\newtheorem{conventionn}[theoremm]{Convention}
\newtheorem{examplee}[theoremm]{Example}
\newtheorem{remarkk}[theoremm]{Remark}
\newtheorem{factt}[theoremm]{Fact}
\newtheorem{exercisee}[theoremm]{Exercise}

\newenvironment{exercise}{\begin{exercisee} \em}{ \end{exercisee}}
\newenvironment{definition}{\begin{definitionn} \em}{ \end{definitionn}}
\newenvironment{theorem}{\begin{theoremm}}{\end{theoremm}}
\newenvironment{lemma}{\begin{lemmaa}}{\end{lemmaa}}
\newenvironment{proposition}{\begin{propositionn} }{\end{propositionn}}
\newenvironment{convention}{\begin{conventionn} \em}{\end{conventionn}}
\newenvironment{remark}{\begin{remarkk} \em}{\end{remarkk}}
\newenvironment{proof}{ {\bf Proof.} }{\  $\Box$ \vspace{.1in} }
\newenvironment{example}{\begin{examplee} \em}{\end{examplee}}
\newenvironment{fact}{\begin{factt}}{\end{factt}}

\title{Separating the basic logics of the basic recurrences}
\author{Giorgi Japaridze\thanks{Supported by 2010 Summer Research Fellowship from Villanova University}}
\date{}
\maketitle

\begin{abstract} This paper shows that, even at the most basic level (namely, in combination with only $\gneg,\mlc,\mld$), the parallel, countable branching and uncountable branching recurrences of computability logic validate different principles.
\end{abstract}

\noindent {\em MSC}: primary: 03B47; secondary: 03B70; 68Q10; 68T27; 68T15. 

\  

\noindent {\em Keywords}: Computability logic; Game semantics; Recurrence operators 

\tableofcontents

\section{Introduction} {\em Computability logic} (CoL) is a long-term project for redeveloping logic on the basis of a constructive game semantics. 
The approach induces a rich collection of logical operators, standing for various natural operations on games. Among those are {\em recurrence operators}, the most basic sorts of which are {\em parallel recurrence} $\pst$,  (uncountable) {\em branching recurrence} $\st$, and {\em countable branching recurrence}  $\bst$. Each recurrence operator $!$ comes with its dual $?$, defied by $?F=\gneg !\gneg F$. The present paper shows that the logical behaviors of these three sorts of recurrences are pairwise distinct --- that is, they validate different principles --- even at the most basic level, namely, in combination with only {\em negation} $\gneg$, {\em parallel conjunction} $\mlc$, and {\em parallel disjunction} $\mld$ (as always, either $\mld$ or $\mlc$, being definable from the other, can be dropped). 

Showing validity or non-validity of various principles in CoL tends to  be far from easy.  This is especially so for principles involving recurrence operators. Recent years (\cite{Japtocl1}-\cite{Cirq}, \cite{Japtcs}-\cite{Japfour}, \cite{Japtowards}-\cite{Ver} and more) have seen rapid and sustained progress in finding sound and complete axiomatizations for many, often quite expressive, fragments of CoL, at both the propositional and the first-order levels. Those fragments, however, have typically been recurrence-free.\footnote{The so called intuitionistic fragment of CoL, studied in \cite{Japjsl,Propint,Ver}, is the only exception. There, however, the usage of recurrence $!$ is limited to the very special form/context $!E\mli F$.} So, it would be accurate to say that, at this point, practically nothing is known about the logical behavior of recurrences, and finding  syntactic descriptions (such as axiomatizations) of the logics induced by them remains among the greatest challenges in the entire  CoL enterprise. The present paper attempts to 
bring some initial light  
 into this otherwise completely dark picture. Its results constitute a necessary first step on the presumably long road of syntactically taming recurrences:  
before even considering looking for axiomatizations, one needs to know whether to expect for those axiomatizations to be common or different for the various sorts of recurrences naturally emerging in game semantics. 

The logics induced by the three recurrences turn out to be separated by the following two principles, which we  call {\bf short production} and {\bf long production}, respectively: 
%\marginpar{n1}
\begin{equation}\label{n1}
P\hspace{2pt}\mlc \hspace{2pt}!\hspace{1pt}(P\mli P\mlc P)\hspace{3pt}\mli \hspace{4pt}!\hspace{1pt}P;\vspace{0pt}
\end{equation}
%\marginpar{n2}
\begin{equation}\label{n2}
P\hspace{2pt}\mlc \hspace{2pt}!\hspace{1pt}(P\mli P\mlc P)\hspace{2pt}\mlc\hspace{2pt} !\hspace{1pt}(P\mld P\mli P)\hspace{3pt}\mli\hspace{4pt} !\hspace{1pt}P.\vspace{7pt}
\end{equation} 
Namely, the situation is as shown in Figure 1, with  {\bf validity} throughout this paper  understood as what CoL calls {\em uniform} (as opposed to the weaker {\em multiform}) validity.\footnote{Extensionally, multiform validity (in most earlier papers on CoL simply called {\em validity}) typically coincides with uniform validity (\cite{Japtocl1}-\cite{Japtocl2},\cite{Japtcs}-\cite{Japseq},\cite{Japtoggling}), but tends to be harder to deal with in completeness proofs, even though a way of turning completeness proofs with respect to uniform validity into completeness proofs with respect to 
multiform validity appears to be more or less standard. Also, in all applications, it is uniform validity that matters, with multiform validity being  of purely theoretical interest. For these reasons, 
in the latest papers on CoL, including the present one, the interest has shifted towards uniform validity, in completeness proofs no longer addressing the question on multiform validity --- at least temporarily so.}

\begin{center} 
\begin{picture}(284,85)

\put(100,75){{when} $!=\pst$}

\put(165,75){{when} $!=\st^{\aleph_0}$}

\put(237,75){{when} $!=\st$}

\put(113,55){{\em valid}}

\put(175,55){{\em invalid}}

\put(245,55){{\em invalid}}

\put(113,35){{\em valid}}

\put(178,35){{\em valid}}

\put(245,35){{\em invalid}}

\put(0,55){short production}

\put(0,35){long production}

\put(120,10){\bf Figure 1}
\end{picture}
\end{center}

 This result is by no means obvious. One could have just as well expected that the differences between the three types of recurrences are too subtle to induce non-identical logics, at least at the $(\gneg,\mlc,\mld,!,?)$-level. For instance, as shown in \cite{Japjsl,Japfour}, the implicative logic induced by all three recurrences $!\in\{\pst,\st,\bst\}$, with implication $E\supset F$ understood as $!E\mli F$, is exactly the implicative fragment of Heyting's intuitionistic calculus. The fact that the seemingly ``almost the same'' operators $\st$ and $\bst$ induce different logics is especially surprising.  

The intended audience for this relatively short (by the standards of CoL) and technical paper is expected to be familiar with the main concepts of CoL, such as those of static games, hard-  and easy-play machines, the operators $\gneg,\mlc,\mld,\pst,\pcost,\st,\cost,\ada,\ade$ (as always, $A\mli B$ is an abbreviation of $\gneg A\mld B$), interpretation, validity, and the related notions. If not, it would be both necessary and sufficient to read the first ten sections of \cite{Japfin}  for a self-contained, tutorial-style introduction. The definition of $\st$ given in \cite{Japfin} is a little bit long and, for that reason, this paper re-introduces this operation, together with its ``countable'' counterpart $\bst$, through a shorter definition. No other operations and concepts will be reintroduced and, again, they are to be understood as defined or explained in \cite{Japfin}.

\section{Recurrence operations: a quick review} Officially, (uncountable) branching recurrence $\st$  was first introduced in \cite{Jap03}, parallel recurrence $\pst$ in \cite{Japic}, and countable branching recurrence $\bst$  in \cite{Japfour}. $\st$ is the author's favorite, as it permits reusing its argument (as a resource) in the strongest algorithmic sense possible, thus allowing us to claim that the compound operation $\st A\mli B$ captures our most general intuition of algorithmically reducing $B$ to $A$.  The weaker $\pst$ stands out as the simplest sort of a recurrence. $\bst$, by its strength strictly between $\pst$ and $\st$, is, in a sense, the strongest of all possible nontrivial weakenings of $\st$. Our interest in $\bst$ is partly also historical. It is related to the apparent fact that $\bst$ is ``equivalent'' to Blass's \cite{Bla92} {\em repetition operator} $R$, the idea of which, in fact, was already present  in \cite{Bla72}, fifteen years before a similar (in the overall logical spirit) idea was materialized in the form of the exponential operator $!$ of linear logic. Here the qualification ``equivalent''  lacks a precise meaning, 
because $\bst$ and $R$  operate in non-identical game-semantical contexts (among the differences is that Blass's games are strict while the CoL games are not), which have never been brought to a common denominator. In a precise yet weaker wording, it is believed that, at least, the logical behaviors of the two operators are indistinguishable. Such a claim was made in \cite{Ver} and, while no proof has been attempted, the present author has hardly any doubts that it is correct.  

As we probably remember, $\pst A$ is defined simply as the infinite $\mlc$-conjunction $A\mlc A\mlc A\mlc\ldots$. 
The operator $\st$ is technically much more involved. 
In semiformal terms, a play of $\st A$  starts as an ordinary play of game $A$. At any time, however, player $\oo$ (the environment) is allowed to make a ``replicative move'', which creates two copies of the current position $\Phi$ of $A$. From that point on, the game turns into two games played in parallel, each continuing 
from position $\Phi$. We use the bits $0$ and $1$ to denote those two threads, which have a common past (position $\Phi$) but possibly diverging futures. Again, at any time, $\oo$ can further branch either thread, creating two copies of the current position in that thread. If thread $0$ was branched, the resulting two threads will be denoted by $00$ and $01$; and if the branched thread was $1$, then the resulting threads will be denoted by $10$ and $11$. And so on: at any time, $\oo$ may split any of the existing threads $w$ into two threads $w0$ and $w1$. Each thread in the eventual run of the game will be thus denoted by a (possibly infinite) bit string. The game is considered won by $\pp$ (the machine) if it wins $A$ in each of the threads; otherwise the winner is $\oo$. 

To each infinite bit string $w$ may thus correspond a separate run of $A$ in thread (represented by) $w$ and, as there are uncountably many infinite bit strings, uncountably many parallel runs of $A$ may be generated when playing $\st A$. Let us call a bit string $w$ {\bf essentially finite} if it contains only a finite number of ``$1$''s; otherwise we say that $w$ is {\bf essentially infinite}. We extend these terms from bit strings to the corresponding threads in the play of $\st A$. The definition of $\st A$ thus requires from $\pp$ to win $A$ in all --- whether they be essentially finite or essentially infinite --- threads. All it takes to turn that definition into a definition of $\bst$ is to relax that requirement and, when determining the winner, only look at essentially finite threads. Since there are only countably many essentially finite bit strings, only countably many runs of $A$ are generated --- more precisely, only countably many runs of $A$ are of relevance --- in $\bst A$. This completes our semiformal definition/explanation of $\st$ and $\bst$.

In fully formal terms, consider a (constant) game $A$. Both $\st A$ and $\bst A$ have the same sets of legal runs.  There are two types of legal moves in (legal) positions  of either game: (1) replicative and (2) non-replicative. To define these, let us agree that by an {\bf actual node}
of a (legal) position $\Phi$ of $\st A$ or $\bst A$ we mean a bit string $w$ such that $w$ is either empty,\footnote{Intuitively, the empty bit string is the name/address of the initial thread; all other threads will be descendants of that thread.} 
or else is $u0$ or $u1$ for some bit string $u$ such that $\Phi$ contains the move $\col{u}$. 
An actual node is said to be a {\bf leaf} iff it is not a proper prefix of any other actual node.\footnote{Intuitively, a leaf 
is the unique individual name of an already existing thread of a play over $A$, while an actual node $w$ which is not a leaf is 
a ``partial'' common name of several already existing threads --- namely, all threads whose individual names look like $wv$ for some $v$.}  A replicative move can only be made by (is only legal for) $\oo$, and such a move in a given position $\Phi$ should be $\col{w}$, where $w$ is a leaf of
$\Phi$.\footnote{The intuitive meaning of move $\col{w}$ is splitting thread $w$ into $w0$ and $w1$, thus ``activating'' these two new nodes/threads.} As for non-replicative moves, they can be made by either player. Such a move by a player $\xx$ in a given position $\Phi$ should be $w.\alpha$, where   $w$ is an actual node of $\Phi$ and $\alpha$ is a move such that, for any infinite bit string $v$, $\alpha$ is a legal move by $\xx$ in position $\Phi^{\preceq wv}$ of $A$.\footnote{The intuitive meaning of such a move $w.\alpha$ is making move $\alpha$ in thread $w$ and all of its (current or future) descendants.} Here, for a run $\Theta$ and a bit string $x$, $\Theta^{\preceq x}$ means the result of deleting from $\Theta$ all moves except those that look like $u.\beta$ for some initial segment $u$ of $x$, and then further deleting the prefix ``$u.$'' from such moves.\footnote{Intuitively, $\Theta^{\preceq x}$ is the run of $A$ that has been played in thread $x$, if such a thread exists (has been  generated); otherwise, $\Theta^{\preceq x}$ is the run of $A$ that has been played in (the unique) existing thread which (whose name, that is) is
some initial segment of $x$.} A legal run $\Gamma$ of $\st A$ is considered won by $\pp$ iff, for every infinite bit string $v$, $\Gamma^{\preceq v}$ is a $\pp$-won run of $A$.  And a legal run $\Gamma$ of $\bst A$ is considered won by $\pp$ iff, for every infinite but essentially finite bit string $v$, $\Gamma^{\preceq v}$ is a $\pp$-won run of $A$. This completes our definition of $\st$ and $\bst$.

What we here semiformally call a {\bf thread} of a (legal) run $\Gamma$ of $\st A$ or $\bst A$ is a generalized leaf. Each thread is named by --- and is usually identified with --- a bit string $w$. When $w$ is finite, saying that it is (it names) a thread means the same as saying that it is a leaf.\footnote{More precisely, if $\Gamma$ is infinite, $w$ should be a leaf of every ``sufficiently long'' finite initial segment $\Phi$ of $\Gamma$.}
 And when $w$ is infinite, saying that it is (it names) a thread means that every finite initial segment  of $w$ is an actual node of some finite initial segment $\Phi$ of $\Gamma$. In other words, a thread is nothing but what the paper \cite{Japfin} more technically refers to as a ``{\em complete branch of the underlying bitstring tree}''. Intuitively, however, and by some abuse of language, when $w$ is a thread, every initial segment $v$ of it we also see as ``the same thread'', because such a $v$ is nothing but a certain ``early stage'' of $w$.

This was a brutally quick review, of course. See \cite{Japfin} for more explanations and illustrations. In our treatment, we shall hardly 
ever rely on the formal definitions of the relevant game operations. Rather, we will be using informal or semiformal explanations and intuitive reasoning. Again, it should be pointed out that the present paper is not meant for a newcomer to the area of CoL.

\section{Parallel recurrence validates both production principles}\label{s3}
%\marginpar{s3}
We first want to set up a uniform (interpretation-independent) winning strategy for short production with $!$ understood as $\pst$. In fact, we can and will do so for the following, more general than (\ref{n1}), form of short production:  
%\marginpar{j26a}
\begin{equation}\label{j26a}
P\hspace{2pt}\mlc\hspace{2pt} !\hspace{1pt}(P\mli P\mlc Q)\hspace{3pt}\mli \hspace{4pt}!\hspace{1pt}Q.
\end{equation}
Writing $\overline{X}$ for $\gneg X$ and eliminating $\mli$, (\ref{j26a}) is rewritten as follows:

%\marginpar{f1}
\begin{equation}\label{f1}
\overline{P}\ \mld \ ?\bigl(P\mlc (\overline{P}\mld \overline{Q})\bigr)\  \mld \ !\hspace{1pt}Q.
\end{equation}
For convenience of references, let us agree that $P_0,P_1,P_2,\ldots$ all mean the same as $P$, and that $Q_1,Q_2,\ldots$ all mean $Q$. Now, remembering that $\pst$ is nothing but an infinite $\mlc$-conjunction, and its dual $\pcost$ is nothing but an infinite $\mld$-disjunction (and that disjunction is associative), (\ref{f1}) can be further rewritten as 

%\marginpar{f2}
\begin{equation}\label{f2}
\overline{P}_0\ \mld \ \bigl(P_1\mlc (\overline{P}_1\mld \overline{Q}_1)\bigr)\ \mld \ \bigl(P_2\mlc (\overline{P}_2\mld \overline{Q}_2)\bigr)\ \mld\ \ldots\ \mld    \ \bigl(Q_1\mlc Q_2\mlc \ldots\bigr).
\end{equation}

Since our strategy does not depend on an interpretation $^*$ applied to (\ref{f2}), we typically omit it (here and later in similar cases) and write, say, $P$ where, strictly speaking, $P^*$ is meant. In other words, with some innocent abuse of concepts, we identify formulas with the games into which they turn after an interpretation is applied to them. 
 
A strategy that solves (\ref{f2}) is rather simple, and is schematically shown in Figure 2. 
\begin{center} 
\begin{picture}(115,193)
\put(29,97){\line(-2,-1){23}}
\put(29,97){\line(0,1){10}}

\put(26,173){\small $\overline{P}_0$}
\put(29,161){\line(-2,-1){23}}
\put(29,161){\line(0,1){10}}

\put(51,129){\line(1,0){56}}
\put(51,129){\line(0,1){10}}
\put(107,129){\line(0,1){10}}
\put(0,141){\small $ P_1\mlc (\overline{P}_1 \mld\overline{Q}_1)$}
\put(102,141){\small $Q_1$}
\put(29,129){\line(-2,-1){23}}
\put(29,129){\line(0,1){10}}

\put(51,97){\line(1,0){56}}
\put(51,97){\line(0,1){10}}
\put(107,97){\line(0,1){10}}
\put(0,109){\small $ P_2\mlc (\overline{P}_2 \mld\overline{Q}_2)$}
\put(102,109){\small $Q_2$}
\put(29,97){\line(-2,-1){23}}
\put(29,97){\line(0,1){10}}

\put(51,65){\line(1,0){56}}
\put(51,65){\line(0,1){10}}
\put(107,65){\line(0,1){10}}
\put(0,77){\small $ P_3\mlc (\overline{P}_3 \mld\overline{Q}_3)$}
\put(102,77){\small $Q_3$}
\put(29,65){\line(-2,-1){10}}
\put(29,65){\line(0,1){10}}

\put(54,47){$\bullet$} 
\put(54,41){$\bullet$} 
\put(54,35){$\bullet$} 

\put(34,10){\bf Figure 2}

\end{picture}
\end{center}

The strategy consists in copycat routines between the pairs of subgames indicated in Figure 2 by arcs; we say that the corresponding two subgames $X$ and $\overline{X}$  are {\bf matched}, or {\bf synchronized} by the strategy.  Synchronizing means mimicking, in $X$, the moves made by the adversary in $\overline{X}$,  and vice versa. If we see our strategy as an EPM (``Easy-Play Machine''), the synchronization is {\em perfect}, in the  sense that the run taking place in $\overline{X}$ is the exact {\em negation}  (all labels reversed) of the run taking place in $X$,  meaning that exactly one of the two subgames will be eventually won by $\pp$. If we see our strategy as an HPM (``Hard-Play Machine''), the synchronization is not necessarily perfect.\footnote{After all, the environment can make any finite number of moves at once while the machine/strategy can make at most one move per computation step.} What is however still guaranteed is that the negation of the run taking place in $\overline{X}$ is a $\oo$-delay (See \cite{Japfin}, Section 5) of the  run taking place in $X$; taking into account that (as always in CoL) the games that we consider are static, this means that at least one of the subgames $X$, $\overline{X}$ will be won by $\pp$, which, in view of the monotonicity of the winning conditions for the relevant game operations, is ``even better than'' when exactly one subgame is won. In view of this observation, for simplicity, here and later we will pretend that a synchronization is always perfect. We may not always be specific about whether the strategy that we consider is an HPM or an EPM, as these two models are equivalent for static games (See \cite{Japfin}, Section 6). This equivalence, in turn, allows us to pretend (usually only implicitly) that the adversary of a given strategy always waits patiently until the strategy permits it to move.

Anyway, a simple combinatorial analysis of the situation convinces us that the strategy represented in Figure 2 wins the game: in view of its matching arrangements, one can see that if one of $Q_i$ ($i\geq 1$) is lost by $\pp$, then either $\overline{P}_0$ or one of $P_j\mlc(\overline{P}_j\mld \overline{Q}_j)$ ($1\leq j\leq i$) is won, and hence so is the overall game.

As for long production, again, we can claim its being validated by $\pst$ in forms more general than the originally given (\ref{n2}). One of such forms is  
%\marginpar{july1a}
\begin{equation}\label{july1a}
P\hspace{2pt}\mlc\hspace{2pt} !\hspace{1pt}(P\mli P\mlc Q)\hspace{2pt}\mlc\hspace{2pt} !\hspace{1pt}(R\mld Q\mli R)\hspace{3pt}\mli \hspace{4pt}!\hspace{1pt}R.
\end{equation}
As shown in \cite{Japfin}, with $!=\pst$, everything provable in affine logic is uniformly valid; and uniform validity is closed under modus ponens. The formula $(\ref{j26a})\mli (\ref{july1a})$ can easily be seen to be provable in affine logic. And, as we already know, its antecedent is uniformly valid. Hence so is its consequent.
  
\section{Branching recurrences do not validate short production}\label{s4}
%\marginpar{s4}

We are going to show that the following instance of short production, with either $!\in\{\st,\bst\}$, is not valid, where $p$ is a binary {\em elementary  letter} (see Section 7 of \cite{Japfin}): 
\[\ade x\ada y\hspace{1pt}p(x,y)\ \mlc \ !\hspace{1pt}\Bigl(\ade x\ada y\hspace{1pt}p(x,y)\mli \bigl(\ade x\ada y\hspace{1pt}p(x,y)\mlc \ade x\ada y\hspace{1pt}p(x,y)\bigr)\Bigr)\  \mli \ \hspace{2pt} !\hspace{1pt}\ade x\ada y\hspace{1pt}p(x,y).
\]

Using $\overline{p}(x,y)$ for $\gneg p(x,y)$, the above formula is rewritten  
 as follows:
%\marginpar{f4}
\begin{equation}\label{f4}
\ada x\ade y\hspace{1pt}\overline{p}(x,y)\ \mld \ ?\Bigl(\ade x\ada y\hspace{1pt}p(x,y)\mlc \bigl(\ada x\ade y\hspace{1pt}\overline{p}(x,y)\mld \ada x\ade y\hspace{1pt}\overline{p}(x,y)\bigr)\Bigr)\  \mld \ !\hspace{1pt}\ade x\ada y\hspace{1pt}p(x,y).
\end{equation}

Let us fix an HPM $\cal H$ as an arbitrary strategy of the machine ($\pp$). We want to construct a {\em counterstrategy} $\cal C$ such that, when the environment ($\oo$) follows it, $\cal H$ loses (\ref{f4}) under an appropriately selected interpretation, meaning that (\ref{f4}) is not uniformly valid (because $\cal H$ was picked arbitrarily). Technically, $\cal C$ is a function that takes  a computation step $i$ of $\cal H$ as an input, and returns a (possibly empty) sequence of moves that the environment should make during step $i$. 

By a {\bf literal} in this section we mean the formula $p(a,b)$ or $\overline{p}(a,b)$ for whatever constants $a,b$. The first literal is said to be {\bf positive}, and the second literal 
is said to be {\bf negative}. When we say that a given positive literal $p(a,b)$ occurs in a given formula, we always mean that it occurs without negation. For instance, the formula $p(a,b)\mlc \overline{p}(c,d)$ contains the literal $p(a,b)$ --- as well as $\overline{p}(c,d)$ --- but not $p(c,d)$. So, strictly speaking, a literal for us is a formula $p(a,b)$ or $\overline{p}(a,b)$ together with some fixed (usually clear from the context) positive occurrence. Unlike a literal, an {\bf atom} always simply means the formula $p(a,b)$ (for some constants $a,b$), no matter where and how it occurs. 
So, for instance, the above formula  $p(a,b)\mlc \overline{p}(c,d)$ {\em does} contain the atom (but not the literal)  $p(c,d)$.
Any two literals $p(a,b)$ and $\overline{p}(a,b)$ (the same $a,b$) are said to be {\bf opposite}. 

Below is a description of the work of $\cal C$. In it, as in the previous section, terminologically we treat formulas as if they were games. The description assumes the context of a particular step of the play/interaction between $\cal H$ and $\cal C$. In that context, a  {\bf fresh} constant means 
a constant (decimal numeral) that has never been chosen by either player for the variables of (\ref{f4}) so far. A {\bf component}, {\bf part} or {\bf subgame}, unless otherwise specified, always means one of (\ref{f4}). We refer to the three disjuncts of (\ref{f4}) as the {\bf recurrence-free component}, the {\bf $?$-component} and the {\bf $!$-component}, respectively. By an {\bf activated} literal we mean one to which the game has already been brought down in the recurrence-free component, or in one of the threads of the $!$-component, or in one of the parts of one of the threads of the $?$-component. For instance, if, by a given time, in one of the threads of the $?$-component, the game $\ade x\ada y\hspace{1pt}p(x,y)\mlc \bigl(\ada x\ade y\hspace{1pt}\overline{p}(x,y)\mld \ada x\ade y\hspace{1pt}\overline{p}(x,y)\bigr)$ has been brought down to $\ada y\hspace{1pt}p(a,y)\mlc \bigl(\overline{p}(b_1,c)\mld \ade y\hspace{1pt}\overline{p}(b_2,y)\bigr)$, then $\overline{p}(b_1,c)$ --- but not necessarily $p(b_1,c)$ --- is an activated literal at that (and any later) time. 
 We implicitly assume that $\cal H$ never makes illegal moves, or otherwise $\cal C$ wins immediately. \vspace{4pt}

We define the counterstrategy $\cal C$ --- or rather describe it in semiformal terms --- through the following four ``prescriptions'':
\begin{description}
  \item[Prescription (i):] At the very beginning of the play, choose the constant $1$ for $x$ in the recurrence-free component, thus bringing it down to $\ade y\hspace{1pt}\overline{p}(1,y)$.
  \item[Prescription (ii):]  Whenever, in any given thread $w$ of the $?$-component, $\cal H$ chooses a constant $a$ for $x$ in the $\ade x\ada y\hspace{1pt}p(x,y)$ part, 
choose fresh constants $b_1,b_2$ ($b_1\not= b_2$) for the two occurrences of $x$ in the $\ada x\ade y\hspace{1pt}\overline{p}(x,y)\mld \ada x\ade y\hspace{1pt}\overline{p}(x,y)$ part of the same thread. Thus, as a result, in thread $w$ we will now have\footnote{I.e., the (sub)game $\add x\adc y\hspace{1pt}p(x,y)\wedge \bigl(\adc x\add y\hspace{1pt}\overline{p}(x,y)\vee \adc x\add y\hspace{1pt}\overline{p}(x,y)\bigr)$ will be brought down to \ldots}
\[\ada y\hspace{1pt}p(a,y)\mlc \bigl(\ade y\hspace{1pt}\overline{p}(b_1,y)\mld \ade y\hspace{1pt}\overline{p}(b_2,y)\bigr).\]
  \item[Prescription (iii):] If and when $\cal H$ chooses a constant $\mathfrak{m}$ for $x$ in the (so far) single thread of the $!$-component,
 split that thread, and choose fresh constants $\mathfrak{n}_1,\mathfrak{n}_2$ ($\mathfrak{n}_1\not= \mathfrak{n}_2$) for $y$ in the two newly emerged threads. Thus, as a result, we will now have two threads in the $!$-component, one containing $p(\mathfrak{m},\mathfrak{n}_1)$ and the other containing $p(\mathfrak{m},\mathfrak{n}_2)$. From this point on, no moves in the $!$-component can or will ever be made again by either player. 
  \item[Prescription (iv):] Suppose, in a given thread $w$ of the $?$-component, by now the game has been brought down to one of the following forms:
\begin{equation}\label{u1}
\ada y\hspace{1pt}p(a,y)\mlc \bigl(\ade y\hspace{1pt}\overline{p}(b,y)\mld F\bigr);
\end{equation}
\begin{equation}\label{u2}
\ada y\hspace{1pt}p(a,y)\mlc \bigl(F\mld \ade y\hspace{1pt}\overline{p}(b,y)\bigr)
\end{equation}
(where $F$ is either $\ade y\hspace{1pt}\overline{p}(b',y)$ or $\overline{p}(b',c')$ for some $b',c'$), and $\cal H$ chooses a constant $c$ for $y$ 
in the $\ade y\hspace{1pt}\overline{p}(b,y)$ subcomponent, such that the literal $p(b,c)$ is (already) activated. Then choose a fresh constant $d$  for $y$ in the $\ada y\hspace{1pt}{p}(a,y)$ part of the same thread. Thus, depending on which of (\ref{u1}), (\ref{u2}) was the case, in thread $w$ we will now have one of the following: 
\[p(a,d)\mlc \bigl(\overline{p}(b,c)\mld F\bigr);\]  
 \[p(a,d)\mlc \bigl(F\mld \overline{p}(b,c)\bigr).\vspace{7pt}\]  
\end{description}

Consider the play of $\cal H$ in the scenario where the environment acts according to counterstrategy $\cal C$.  There are two possibilities to be looked at separately.

One possibility is that it never comes to acting according to Prescription (iii).  This means that the $!$ component remains unchanged throughout the play. Observe that then it never comes to acting according to Prescription (iv) either, because Prescription (iii) is the only place where 
the {\em first-ever} positive activated literal can emerge --- a literal whose presence is required in Prescription (iv). In this case, we choose  an interpretation that makes every atom $p(a,b)$ true. A rather straightforward analysis of the situation convinces us that $\cal H$ loses (\ref{f4}) under this interpretation.

The other possibility, on which we focus throughout the rest of this section, is that, at some point, it comes to $\cal C$ acting according to Prescription (iii).  
Let us fix  $\mathfrak{m},\mathfrak{n}_1,\mathfrak{n}_2$ as the (unique) constants from Prescription (iii). 

Let us say that a literal $p(a,d)$ is a {\bf threadmate} of a literal $\overline{p}(b,c)$ iff there is a formula/game 
$p(a,d)\mlc\bigl(\overline{p}(b,c)\mld F\bigr)$ or $p(a,d)\mlc\bigl(F\mld \overline{p}(b,c)\bigr)$\footnote{In either case, $F$ is either 
$\ade y\hspace{1pt}\overline{p}(b',y)$ or $\overline{p}(b',c')$ for some $b',c'$.} to which the game has been brought down\footnote{In this context meaning ``has been brought down at some point in the play''. The same applies to our usage of ``activated'' and similar terms.} in one of the threads of the ?-component.

We define a {\bf chain} as a nonempty finite sequence of activated literals of the following form:   
\begin{equation}\label{chain}
p(a_1,b_1),\ \overline{p}(a_1,b_1),\ldots, \ p(a_{n-1},b_{n-1}),\ \overline{p}(a_{n-1},b_{n-1}), \ p(a_n,b_n),\ \overline{p}(a_n,b_n)
\end{equation} 
($n\geq 1$), where (in addition to what can be seen from the above display --- namely, that every literal at an odd position is positive and is followed by the opposite literal), for each $i$ with $1\leq i<n$, the literal
$p(a_{i+1},b_{i+1})$ is a threadmate of the literal $\overline{p}(a_i,b_i)$.

What we call a {\bf semichain} 
\begin{equation}\label{semichain}
p(a_1,b_1),\ \overline{p}(a_1,b_1),\ \ldots, \ p(a_{n-1},b_{n-1}),\ \overline{p}(a_{n-1},b_{n-1}), \ p(a_n,b_n)
\end{equation}
satisfies exactly the same conditions as a chain, with the only difference that the last element $\overline{p}(a_n,b_n)$ is not present. 
Thus, chains are even-length while semichains are odd-length. Of course, if (\ref{chain}) is a chain, then (\ref{semichain}) is a semichain. But not necessarily vice versa: if the literal $\overline{p}(a_n,b_n)$ is not activated, then (\ref{chain}) is not a chain even if (\ref{semichain}) is a semichain.

We say that a constant $a$ is {\bf activated} iff it has been chosen for $x$ or $y$ at some point by either player in any (thread of any) part of the play.   Using consecutive positive integers $1,2,\ldots$ for the computation steps of $\cal H$, by the {\bf activation time} $\atime(a)$ of such a constant we mean the earliest  computation step of $\cal H$ during which $a$ was first chosen by the corresponding player.  Remember that, in the HPM model, the machine can make at most one move during a given computation step, while the environment can make any finite number of moves. We assume that the two moves that $\cal C$ makes according to Prescription (ii) --- choosing $b_1$ and choosing $b_2$ --- happen during the same step, so that $\atime(b_1)=\atime(b_2)$.

Where $i\in\{1,2\}$, a  {\bf $p(\mathfrak{m},\mathfrak{n}_i)$-headed chain} is a chain whose first literal is $p(\mathfrak{m},\mathfrak{n}_i)$. When we do not want to be specific about whether $i=1$ or $i=2$, we simply say ``{\bf a headed chain}''.  The same terminology extends from chains to semichains. 

\begin{lemma}
\label{j27a}
%\marginpar{j27a}
Suppose the following are headed semichains: 
\[p(a_1,b_1),\ \overline{p}(a_1,b_1),\ \ldots,\ p(a_{n-1},b_{n-1}),\ \overline{p}(a_{n-1},b_{n-1}), \ p(a_n,b_n);\]
\[p(a'_1,b'_1),\ \overline{p}(a'_1,b'_1),\ \ldots,\ p(a'_{n-1},b'_{n-1}),\ \overline{p}(a'_{n-1},b'_{n-1}), \ p(a'_n,b'_n).\]
 Then  $a_1=a'_1,\ldots,a_n=a'_n$ and $\atime(a_1)>\ldots >\atime(a_n)$. 
\end{lemma}

\begin{proof} Induction on $n$. For $n=1$, the statement of the lemma  is immediate, because $a_1=a'_1=\mathfrak{m}$. Now, suppose $n>1$. By the definition of a semichain, there is a thread $w$ in the $?$-component where, for some $F$, at some point, the game was brought down to 
\label{j29a}
\begin{equation}\label{j29a}
p(a_{n},b_n)\mlc\bigl(\overline{p}(a_{n-1},b_{n-1})\mld F\bigr) 
\end{equation}
 (or to $p(a_{n},b_n)\mlc\bigl(F\mld \overline{p}(a_{n-1},b_{n-1})\bigr) $, but this case is similar); there is also a thread $w'$ where, for some $F'$, at some point,
 the game was brought down to 
\label{j29b}
\begin{equation}\label{j29b}
p(a'_{n},b'_n)\mlc\bigl(\overline{p}(a'_{n-1},b'_{n-1})\mld F'\bigr) 
\end{equation}
 (or to $p(a'_{n},b'_n)\mlc\bigl(F'\mld \overline{p}(a'_{n-1},b'_{n-1})\bigr) $, but this case is similar). Let us try to trace the history of how (\ref{j29a}) and (\ref{j29b}) emerged in the corresponding threads. At the beginning, both threads $w$ and $w'$ --- just like all threads of the $?$-component --- had the formula/game 
$\ade x\ada y\hspace{1pt}p(x,y)\mlc\bigl(\ada x\ade y\hspace{1pt}\overline{p}(x,y)\mld\ada x\ade y\hspace{1pt}\overline{p}(x,y) \bigr) $. Some time later, the game in $w$  became  
\label{j29c}
\begin{equation}\label{j29c}
\ada y\hspace{1pt}p(a_n,y)\mlc\bigl(\ada x\ade y\hspace{1pt}\overline{p}(x,y)\mld\ada x\ade y\hspace{1pt}\overline{p}(x,y) \bigr), 
\end{equation}
 and the game in $w'$ became 
\label{j29d}
\begin{equation}\label{j29d}
\ada y\hspace{1pt}p(a'_n,y)\mlc\bigl(\ada x\ade y\hspace{1pt}\overline{p}(x,y)\mld\ada x\ade y\hspace{1pt}\overline{p}(x,y) \bigr).
\end{equation}
In view of Prescription (ii), the former event was followed by $\cal C$ turning (\ref{j29c}) into 
\label{j29e}
\begin{equation}\label{j29e}
\ada y\hspace{1pt}p(a_n,y)\mlc\bigl(\ade y\hspace{1pt}\overline{p}(a_{n-1},y)\mld\ade y\hspace{1pt}\overline{p}(k,y) \bigr)
\end{equation}
for some $k$, and the latter event 
was followed by $\cal C$ turning (\ref{j29d}) into 
\label{j29f}
\begin{equation}\label{j29f}
\ada y\hspace{1pt}p(a'_n,y)\mlc\bigl(\ade y\hspace{1pt}\overline{p}(a'_{n-1},y)\mld\ade y\hspace{1pt}\overline{p}(k',y) \bigr)
\end{equation}
for some $k'$. Note that  $a_{n-1}$  (as well as $a'_{n-1}$, $k$ and $k'$) had to be fresh, for otherwise $\cal C$ would not have chosen it. So, $\atime(a_{n-1})>\atime(a_n)$. We further claim that, by the time (\ref{j29e}) and (\ref{j29f}) emerged, the threads $w$ and $w'$ had not yet diverged; that is, (\ref{j29e}) and (\ref{j29f}) are the same and hence $a_n=a'_n$. Indeed, deny this. Note that then, as constants chosen by $\cal C$ in different threads, $a_{n-1}$ and $a'_{n-1}$ would be different from each other, because $\cal C$ always selects fresh constants. But, by the induction hypothesis, $a_{n-1}=a'_{n-1}$, which is a contradiction. 
\end{proof}

Let us say that a positive literal $p(a,b)$ is {\bf reachable} iff it appears in (or, equivalently, is the last element of) some headed semichain. 

By the {\bf activation time} of an activated positive literal $p(a,d)$ we mean the time at which the constant $d$ was (first) chosen for $y$ by $\cal C$. 
In other words, this is  the time (computation step of $\cal H$) at which the literal  $p(a,d)$ first emerged in the overall play. 

\begin{lemma}\label{j29v}
%\marginpar{j29v}
Every positive activated literal is reachable. 
\end{lemma}

\begin{proof} Consider an arbitrary positive activated literal $p(a,d)$. We proceed by induction on the activation time of $p(a,d)$. Among all positive activated literals,   $p(\mathfrak{m},\mathfrak{n}_1)$ and $p(\mathfrak{m},\mathfrak{n}_2)$ obviously have the smallest activation time; and 
they are reachable.   
 
Now consider any other activated positive literal $p(a,d)$. Obviously it could have emerged in the play only according to Prescription (iv).  Namely, at some point, in some thread of the $?$-component, we had $\ada y\hspace{1pt}p(a,y)\mlc\bigl(\ade y\hspace{1pt}\overline{p}(b,y)\mld F\bigr)$ (or $\ada y\hspace{1pt}p(a,y)\mlc\bigl(F\mld\ade y\hspace{1pt}\overline{p}(b,y)\bigr)$, but this case is similar), and the event that triggered the application of Prescription (iv) was that the above became  $\ada y\hspace{1pt}p(a,y)\mlc\bigl(\overline{p}(b,c)\mld F\bigr)$,  where the (positive) literal $p(b,c)$ had already been activated. To this event, $\cal C$ responded by further bringing the game in the thread down to $p(a,d)\mlc\bigl(\overline{p}(b,c)\mld F\bigr)$ for a fresh $d$, thus ``activating'' $p(a,d)$. The activation time of  $p(b,c)$ is thus smaller than that of $p(a,d)$. Hence,
by the induction hypothesis,  $p(b,c)$ is reachable. But $p(a,d)$ is a threadmate of $\overline{p}(b,c)$. Hence $p(a,d)$ is also reachable: a semichain ending in   $p(a,d)$ is obtained from the semichain ending in $p(b,c)$ by appending to it the two literals $\overline{p}(b,c)$ and $p(a,d)$.
\end{proof}

We say that a constant $a$ is {\bf reachable} iff, for some $b$, the positive literal $p(a,b)$ is reachable. 

\begin{lemma}\label{j29w}
%\marginpar{j29w}
Suppose $b$  and $c$ are two distinct reachable constants. Then $\atime(b)\not=\atime(c)$. 
\end{lemma}

\begin{proof} Assume $b\not= c$ are reachable constants. In view of Lemma \ref{j27a}, with some thought, one can see that there is a finite sequence 
$a_1,\ldots,a_n$ of constants such that $a_1=\mathfrak{m}$, $\atime(a_1)>\ldots >\atime(a_n)$ and any reachable constant is among $a_1,\ldots,a_n$ (hint: consider a longest headed semichain). So, we have $b=a_i$ and $c=a_j$ for some $1\leq i\not= j\leq n$. And therefore $\atime(b)<\atime(c)$ or $\atime(c)>\atime(b)$, i.e., $\atime(b)\not=\atime(c)$.
\end{proof}

\begin{lemma}\label{j29x}
%\marginpar{j29x}
Suppose 
\[p(a_1,b_1),\ \overline{p}(a_1,b_1),\ \ldots,\ p(a_n,b_n),\ \overline{p}(a_n,b_n)\]
is a $p(\mathfrak{m},\mathfrak{n}_1)$-headed chain, and
\[p(a_1,b'_1),\ \overline{p}(a_1,b'_1),\ \ldots,\ p(a_n,b'_n),\ \overline{p}(a_n,b'_n)\]
is a $p(\mathfrak{m},\mathfrak{n}_2)$-headed chain.  Then $b_1\not=b'_1,\ldots,b_n\not=b'_n$. 
\end{lemma}

\begin{proof} We show that $b_i\not= b'_i$ by induction on $i\in\{1,\ldots,n\}$. For $i=1$ this is immediate, because $b_1=\mathfrak{n}_1$ and $b'_1=\mathfrak{n}_2$. 

Now consider any $i\in\{2,\ldots,n\}$. Let us trace the history of the two threads $w$ and $w'$ of the $?$-component in which $p(a_i,b_i)$ and $p(a_i,b'_i)$ were  ``activated'' (i.e., first emerged), respectively.  Originally, in both threads we had
\[\ade x\ada y\hspace{1pt}p(x,y)\mlc \bigl(\ada x\ade y\hspace{1pt}\overline{p}(x,y)\mld \ada x\ade y\hspace{1pt}\overline{p}(x,y)\bigr),\]
which later evolved to 
%\marginpar{jul3a}
\begin{equation}\label{jul3a}
\ada y\hspace{1pt}p(a_i,y)\mlc \bigl(\ada x\ade y\hspace{1pt}\overline{p}(x,y)\mld \ada x\ade y\hspace{1pt}\overline{p}(x,y)\bigr).
\end{equation}
In both threads,  $\cal C$'s response according to Prescription (ii) brought (\ref{jul3a}) down to 
%\marginpar{jul2a}
\begin{equation}\label{jul2a}
\ada y\hspace{1pt}p(a_i,y)\mlc \bigl(\ade y\hspace{1pt}\overline{p}(a_{i-1},y)\mld \ade y\hspace{1pt}\overline{p}(c,y)\bigr)
\end{equation}
(or $\ada y\hspace{1pt}p(a_i,y)\mlc \bigl(\ade y\hspace{1pt}\overline{p}(c,y)\mld \ade y\hspace{1pt}\overline{p}(a_{i-1},y)\bigr)$, but this case is similar) for some $c$ different from $a_{i-1}$. Here we see that at the time of the action that resulted in (\ref{jul2a}), the two threads $w$ and $w'$ were not separated yet, that is, whatever we have said so far, was happening in the common ancestor of the two threads. This is so because, otherwise, $\cal C$ would have chosen distinct $a_{i-1}$s in the two threads. 

But the two threads had to diverge at some point, because otherwise $b_{i-1}$ and $b'_{i-1}$, chosen later by $\cal H$ for $y$ in $\ade y\hspace{1pt}\overline{p}(a_{i-1},y)$, would have to be the same, which, however, is not the case by the induction hypothesis. 
 
If the two threads diverged  before in any way modifying (\ref{jul2a}), then $b_i$ and $b'_i$, as constants chosen later by $\cal C$ in different threads, are different, and we are done. 

Now suppose  some change happened in (\ref{jul2a}) before the two threads diverged.  What could have been such a change?  $\cal C$ would not have moved in (\ref{jul2a})  until $\cal H$ had made a move there first. But $\cal H$ could not have moved within the $\ade y\hspace{1pt}\overline{p}(a_{i-1},y)$ part of (\ref{jul2a}) until the threads diverged (because in $w$ it had to select $b_{i-1}$ for $y$ while in $w'$ select $b'_{i-1}$; these two, by the induction hypothesis, are distinct). So, the only possible event is that $\cal H$ moved   
within the $\ade y\hspace{1pt}\overline{p}(c,y)$ part of (\ref{jul2a}), namely, brought that part down to $\overline{p}(c,d)$ for some $d$. According to Prescription (iv),  $\cal C$ responds to such a move only if the literal $p(c,d)$ is (already) activated, which, by Lemma \ref{j29v}, is the same as to say that $p(c,d)$ is reachable. But if $p(c,d)$ is reachable, then so is $c$; and, of course, $a_{i-1}$ is also reachable. Therefore, by Lemma \ref{j29w}, $\atime(c)\not=\atime(a_{i-1})$. This is a contradiction, because $a_{i-1}$ and $c$ were ``activated'' simultaneously when $\cal C$ brought (\ref{jul3a}) down to (\ref{jul2a}). Thus, $\cal C$ did not respond to $\cal H$'s action, and what we now have in the two, not-yet-diverged threads $w$ and $w'$ is 
%\marginpar{jul2b}
\begin{equation}\label{jul2b}
\ada y\hspace{1pt}p(a_i,y)\mlc \bigl(\ade y\hspace{1pt}\overline{p}(a_{i-1},y)\mld \overline{p}(c,d)\bigr).
\end{equation}
Now, the only next event in the evolution of the two threads is that they, at last, diverge (so that $\cal H$ can choose the different constants $b_{i-1}$ and $b'_{i-1}$ for $y$ in $\ade y\hspace{1pt}\overline{p}(a_{i-1},y)$). But once the threads diverge, as noted earlier, $\cal C$ will later choose different constants $b_i$ and $b'_i$ for $y$ in the $\ada y\hspace{1pt}p(a_i,y)$ component of the two threads. Showing that $b_i\not=b'_i$ was exactly our goal.  
\end{proof}

Let us say that a chain is {\bf complete} iff its last literal is  $\overline{p}(1,b)$ for some $b$. 
 
\begin{lemma}\label{j29y}
%\marginpar{j29y}
Either there is no complete $p(\mathfrak{m},\mathfrak{n}_1)$-headed chain, or there is no complete $p(\mathfrak{m},\mathfrak{n}_2)$-headed chain.  
\end{lemma}

\begin{proof}
Assume, for a contradiction, that there is a complete $p(\mathfrak{m},\mathfrak{n}_1)$-headed chain 
\[p(a_1,b_1), \ \overline{p}(a_1,b_1), \ \ldots, \ p(a_n,b_n), \ \overline{p}(a_n,b_n)\] 
and also there is a 
complete $p(\mathfrak{m},\mathfrak{n}_2)$-headed chain 
\[p(a'_1,b'_1),\  \overline{p}(a'_1,b'_1),\ \ldots,\ p(a'_{n'},b'_{n'}),\  \overline{p}(a'_{n'},b'_{n'}).\]
Without loss of generality here we may assume that $n\leq n'$. According to Lemma \ref{j27a}, we have $a_{1}=a'_{1},\ldots,a_{n}=a'_{n}$.  And, by Lemma \ref{j29x}, $b_n\not= b'_n$. Thus, as $a_n=1$, we have two non-identical activated literals $\overline{p}(1,b_n)$ and  $\overline{p}(1,b'_n)$. This is however impossible. It is impossible because the constant $1$ was chosen by $\cal C$ only in the recurrence-free component (when following Prescription (i)) which, as a result, was brought down to $\ade y\hspace{1pt}\overline{p}(1,y)$; and, since that component is recurrence-free and thus cannot be replicated, $\cal H$ would not have a chance to make two different choices $b_n$ and $b'_n$ for $y$ there to further bring it down to $\overline{p}(1,b_n)$ and $\overline{p}(1,b'_n)$. 
\end{proof}

To complete our proof, pick an $i\in\{1,2\}$ such that there is no complete $p(\mathfrak{m},\mathfrak{n}_i)$-headed chain (the existence of such an $i$ is guaranteed by Lemma \ref{j29y}). Let us say that an atom is {\bf $i$-reachable} iff it appears in some $p(\mathfrak{m},\mathfrak{n}_i)$-headed semichain. 
We choose an interpretation that makes all $i$-reachable atoms false, and makes all other atoms true. It is left to the reader to convince himself or herself that, under this interpretation,  (\ref{f4}) is lost by $\cal H$.  In this exercise, whether $!$ means $\bst$ or $\st$ is of no relevance.

\section{Countable branching recurrence validates long production}

In this section we are going to show that, with $!=\bst$, long production is valid in the strong form of (\ref{july1a}), which we rewrite as  
\begin{equation}\label{e1}
\overline{\cal P}\ \mld \ ?\bigl(P\mlc (\overline{P}\mld \overline{Q})\bigr)\ \mld\  ?\bigl((R\mld Q)\mlc \overline{R}\bigr)\ \mld \ !\hspace{1pt}{\cal R}.
\end{equation}
Here we have used the calligraphic $\cal R$ for one of the two occurrences of $R$ in order to differentiate it from the other occurrence. 
Similarly for $\cal P$. This is merely for readability. 

We refer to the four disjuncts of (\ref{e1}) as the {\bf recurrence-free component}, the {\bf left $?$-component}, the {\bf right $?$-component} and the 
{\bf $!$-component}, respectively. As before, we see formulas as games. Some other earlier terminology and conventions may apply as well. 
 
For our purposes, we want to agree on a simplified way of schematically representing different stages of a play over (\ref{e1}). We explain this way in a semiformal fashion. 
Initially, both $?$-components and the $!$-component have a single thread $\epsilon$ ($\epsilon$ stands for the {\bf empty bit string}). To indicate this, we use $\epsilon$ as a subscript and, after omitting the external disjunction symbols as well as $?$ and $!$, we rewrite (\ref{e1}) as 
\begin{equation}\label{e2}
\overline{\cal P} \hspace{20pt}   P_\epsilon\mlc (\overline{P}_\epsilon \mld\overline{Q}_\epsilon) \hspace{20pt}  (R_\epsilon\mld Q_\epsilon)\mlc \overline{R}_\epsilon\hspace{20pt}  {\cal R}_\epsilon .
\end{equation}
Our purported uniform solution/strategy for (\ref{e1}) --- let us call that strategy $\cal K$ --- makes two initialization moves consisting in replicating the (so far the only) thread $\epsilon$ of both $?$-components. This  results in the position that we represent as 

\begin{equation}\label{e3}
\overline{\cal P} \hspace{20pt}   P_0\mlc (\overline{P}_0 \mld\overline{Q}_0)\hspace{10pt}P_1\mlc (\overline{P}_1 \mld\overline{Q}_1) \hspace{20pt} (R_0\mld Q_0)\mlc \overline{R}_0\hspace{10pt}  (R_1\mld Q_1)\mlc \overline{R}_1\hspace{20pt}  {\cal R}_\epsilon .
\end{equation}
Here we see the subscripts $0$ and $1$ because these are (the names of) the threads into which $\epsilon$ turns after it is split. 
 
After the above initialization moves, $\cal K$ establishes synchronization between the following pairs of subgames: $(\overline{\cal P},P_0)$, $(\overline{Q}_0,Q_0)$ and $(\overline{R}_0,{\cal R}\epsilon)$. This arrangement is shown in Figure 3, with synchronizations indicated by arcs.

\begin{center} 
\begin{picture}(434,119)
\put(57,87){\line(-2,-1){23}}
\put(57,87){\line(0,1){10}}

\put(79,55){\line(1,0){261}}
\put(79,55){\line(0,1){10}}
\put(340,55){\line(0,1){10}}

\put(363,55){\line(1,0){30}}
\put(363,55){\line(0,1){10}}
\put(393,55){\line(0,1){10}}

\put(54,99){\small $\overline{\cal P}$}
\put(28,67){\small $ P_0\mlc (\overline{P}_0 \mld\overline{Q}_0)$}
\put(310,67){\small $(R_0\mld Q_0)\mlc \overline{R}_0$}
\put(389,67){\small ${\cal R}_\epsilon$}

\put(28,35){\small $ P_1\mlc (\overline{P}_1 \mld\overline{Q}_1)$}
\put(310,35){\small $(R_1\mld Q_1)\mlc \overline{R}_1$}

\put(198,10){\bf Figure 3}

\end{picture}
\end{center}

If the environment does not make any replicative moves in the $!$-component, the situation represented by Figure 3 will persist throughout the rest of the game. Of course, $\overline{\cal P}$ will probably no longer be the original $\overline{\cal P}$ at later stages of the play, but what matters is that, whatever game the original $\overline{\cal P}$ evolves to (which we continue denoting by $\overline{\cal P}$), it will essentially be the negation of to whatever game the original $P_0$ evolves; ``essentially'' in the sense explained in Section \ref{s3}, which guarantees that at least one of the two games will be eventually won by $\cal K$ and --- again as in Section \ref{s3} --- we can safely pretend that exactly one of them will be in fact won.
Similarly for the other components of the game shown in Figure 3. A straightforward analysis of the situation in the present scenario (the scenario where the environment does not make replications in the $!$-component) shows that $\cal K$ wins as desired. In this analysis, the presence of the two subgames $ P_1\mlc (\overline{P}_1 \mld\overline{Q}_1)$ and $(R_1\mld Q_1)\mlc \overline{R}_1$ displayed at the bottom of Figure 3 is irrelevant. As we are going to see, the role of these two is to maintain ``fresh'' copies of the original  $ P\mlc (\overline{P} \mld\overline{Q})$ and $(R\mld Q)\mlc \overline{R}$. The same applies to the $\overline{P}_0$ and $R_0$ components. Here and later, by ``fresh'' we mean that $\cal K$ has not made any moves in these (sub)games. Of course, $\cal K$ has no way to prevent the environment from making moves in these subgames. But such moves are harmless in the sense that they cannot create any problems for $\cal K$ later if it decides to start synchronizing such a ``fresh'' subgame $X$ with another ``fresh'' subgame $\overline{X}$. The reason, again as explained in Section \ref{s3}, is that we deal with static games (that is, for any interpretation $^*$, the games $P^*,Q^*,R^*$ are static). 

We continue our description of the strategy $\cal K$. The case of the environment making no replications in the $!$-component has been already fully covered. Now, suppose the environment replicates the thread $\epsilon$ of the $!$-component (at this point, $\epsilon$ is the only thread there). That is, the environment splits the ${\cal R}_\epsilon$ component into two copies ${\cal R}_0$ and ${\cal R}_1$. In response, $\cal K$ replicates the $(R_0\mld Q_0)\mlc \overline{R}_0$ component, turning it into two child copies $(R_{00}\mld Q_{00}) \mlc\overline{R}_{00}$
and $(R_{01}\mld Q_{01}) \mlc\overline{R}_{01}$. It does the same with   
$P_1\mlc (\overline{P}_1\mld \overline{Q}_1)$ and  $(R_1\mld Q_1)\mlc \overline{R}_1$. Figure 4 shows all components that we will be dealing with from now on, and also shows the synchronization arrangements that $\cal K$ will maintain. The ``$\times$'' under $Q_{01}$ indicates that, from now on, this component is ``wasted'' in the sense that it will not and cannot be synchronized with anything.  

\begin{center} \begin{picture}(410,151)

\put(57,119){\line(-2,-1){23}}
\put(57,119){\line(0,1){10}}

\put(79,87){\line(1,0){257}}
\put(79,87){\line(0,1){10}}
\put(336,87){\line(0,1){10}}

\put(363,87){\line(1,0){30}}
\put(363,87){\line(0,1){10}}
\put(393,87){\line(0,1){10}}

\put(54,131){\small $\overline{\cal P}$}
\put(28,99){\small $ P_0\mlc (\overline{P}_0 \mld\overline{Q}_0)$}
\put(303,99){\small $(R_{00}\mld Q_{00})\mlc \overline{R}_{00}$}
\put(389,99){\small ${\cal R}_0$}

\put(56,87){\line(0,1){10}}
\put(56,87){\line(-5,-2){30}}

\put(363,55){\line(1,0){30}}
\put(363,55){\line(0,1){10}}
\put(393,55){\line(0,1){10}}

\put(281,55){\line(1,0){30}}
\put(281,55){\line(0,1){10}}
\put(311,55){\line(0,1){10}}

\put(79,55){\line(1,0){174}}
\put(79,55){\line(0,1){10}}
\put(253,55){\line(0,1){10}}

\put(20,67){\small $ P_{10}\mlc (\overline{P}_{10} \mld\overline{Q}_{10})$}
\put(303,67){\small $(R_{01}\mld Q_{01})\mlc \overline{R}_{01}$}
\put(389,67){\small ${\cal R}_1$}

\put(332,59){\small $\times$}

\put(220,67){\small $(R_{10}\mld Q_{10})\mlc \overline{R}_{10}$}

\put(20,35){\small $ P_{11}\mlc (\overline{P}_{11} \mld\overline{Q}_{11})$}
\put(220,35){\small $(R_{11}\mld Q_{11})\mlc \overline{R}_{11}$}

\put(198,10){\bf Figure 4}
\end{picture}
\end{center}

Again, if the environment makes no further replications, then the synchronization shown in Figure 4 obviously guarantees a win for $\cal K$. Let us now say ${\cal R}_1$ is replicated (the other possibility would be replicating ${\cal R}_0$). $\cal K$'s reaction is replicating 
 $(R_{01}\mld Q_{01})\mlc \overline{R}_{01}$,  
$(R_{10}\mld Q_{10})\mlc \overline{R}_{10}$, 
$ P_{11}\mlc (\overline{P}_{11} \mld\overline{Q}_{11})$ and $(R_{11}\mld Q_{11})\mlc \overline{R}_{11}$, followed by the synchronization arrangements shown in Figure 5.

\begin{center} \begin{picture}(410,183)

\put(57,151){\line(-2,-1){23}}
\put(57,151){\line(0,1){10}}

\put(79,119){\line(1,0){257}}
\put(79,119){\line(0,1){10}}
\put(336,119){\line(0,1){10}}

\put(363,119){\line(1,0){30}}
\put(363,119){\line(0,1){10}}
\put(393,119){\line(0,1){10}}

\put(54,163){\small $\overline{\cal P}$}
\put(28,131){\small $ P_0\mlc (\overline{P}_0 \mld\overline{Q}_0)$}
\put(303,131){\small $(R_{00}\mld Q_{00})\mlc \overline{R}_{00}$}
\put(389,131){\small ${\cal R}_0$}

\put(56,119){\line(0,1){10}}
\put(56,119){\line(-5,-2){30}}

\put(51,87){\line(0,1){10}}
\put(51,87){\line(-5,-2){30}}

\put(363,87){\line(1,0){30}}
\put(363,87){\line(0,1){10}}
\put(393,87){\line(0,1){10}}

\put(273,87){\line(1,0){31}}
\put(273,87){\line(0,1){10}}
\put(304,87){\line(0,1){10}}

\put(79,87){\line(1,0){163}}
\put(79,87){\line(0,1){10}}
\put(242,87){\line(0,1){10}}

\put(20,99){\small $ P_{10}\mlc (\overline{P}_{10} \mld\overline{Q}_{10})$}
\put(296,99){\small $(R_{010}\mld Q_{010})\mlc \overline{R}_{010}$}
\put(389,99){\small ${\cal R}_{10}$}

\put(330,91){\small $\times$}

\put(205,99){\small $(R_{100}\mld Q_{100})\mlc \overline{R}_{100}$}

\put(363,55){\line(1,0){30}}
\put(363,55){\line(0,1){10}}
\put(393,55){\line(0,1){10}}

\put(273,55){\line(1,0){31}}
\put(273,55){\line(0,1){10}}
\put(304,55){\line(0,1){10}}

\put(180,55){\line(1,0){33}}
\put(180,55){\line(0,1){10}}
\put(213,55){\line(0,1){10}}

\put(76,55){\line(1,0){73}}
\put(76,55){\line(0,1){10}}
\put(149,55){\line(0,1){10}}

\put(10,67){\small $ P_{110}\mlc (\overline{P}_{110} \mld\overline{Q}_{110})$}
\put(296,67){\small $(R_{011}\mld Q_{011})\mlc \overline{R}_{011}$}
\put(389,67){\small ${\cal R}_{11}$}

\put(330,59){\small $\times$}
\put(239,59){\small $\times$}
\put(205,67){\small $(R_{101}\mld Q_{101})\mlc \overline{R}_{101}$}

\put(112,67){\small $(R_{110}\mld Q_{110})\mlc \overline{R}_{110}$}

\put(10,35){\small $ P_{111}\mlc (\overline{P}_{111} \mld\overline{Q}_{111})$}
\put(112,35){\small $(R_{111}\mld Q_{111})\mlc \overline{R}_{111}$}

\put(198,10){\bf Figure 5}
\end{picture}
\end{center}

Assume that, next time, the environment replicates ${\cal R}_0$. The situation resulting from $\cal K$'s reaction is shown in Figure 6.

\begin{center} \begin{picture}(402,215)

\put(57,183){\line(-2,-1){23}}
\put(57,183){\line(0,1){10}}

\put(79,151){\line(1,0){254}}
\put(79,151){\line(0,1){10}}
\put(333,151){\line(0,1){10}}

\put(363,151){\line(1,0){30}}
\put(363,151){\line(0,1){10}}
\put(393,151){\line(0,1){10}}

\put(54,195){\small $\overline{\cal P}$}
\put(28,163){\small $ P_0\mlc (\overline{P}_0 \mld\overline{Q}_0)$}
\put(297,163){\small $(R_{000}\mld Q_{000})\mlc \overline{R}_{000}$}
\put(389,163){\small ${\cal R}_{00}$}

\put(56,151){\line(0,1){10}}
\put(56,151){\line(-5,-2){30}}

\put(51,119){\line(0,1){10}}
\put(51,119){\line(-5,-2){30}}

\put(363,119){\line(1,0){30}}
\put(363,119){\line(0,1){10}}
\put(393,119){\line(0,1){10}}

\put(273,119){\line(1,0){31}}
\put(273,119){\line(0,1){10}}
\put(304,119){\line(0,1){10}}

\put(79,119){\line(1,0){163}}
\put(79,119){\line(0,1){10}}
\put(242,119){\line(0,1){10}}

\put(20,131){\small $ P_{10}\mlc (\overline{P}_{10} \mld\overline{Q}_{10})$}
\put(296,131){\small $(R_{010}\mld Q_{010})\mlc \overline{R}_{010}$}
\put(389,131){\small ${\cal R}_{10}$}

\put(330,123){\small $\times$}

\put(205,131){\small $(R_{100}\mld Q_{100})\mlc \overline{R}_{100}$}

\put(363,87){\line(1,0){30}}
\put(363,87){\line(0,1){10}}
\put(393,87){\line(0,1){10}}

\put(273,87){\line(1,0){31}}
\put(273,87){\line(0,1){10}}
\put(304,87){\line(0,1){10}}

\put(180,87){\line(1,0){33}}
\put(180,87){\line(0,1){10}}
\put(213,87){\line(0,1){10}}

\put(76,87){\line(1,0){73}}
\put(76,87){\line(0,1){10}}
\put(149,87){\line(0,1){10}}

\put(10,99){\small $ P_{110}\mlc (\overline{P}_{110} \mld\overline{Q}_{110})$}
\put(296,99){\small $(R_{011}\mld Q_{011})\mlc \overline{R}_{011}$}
\put(389,99){\small ${\cal R}_{11}$}

\put(330,91){\small $\times$}
\put(239,91){\small $\times$}
\put(331,59){\small $\times$}
\put(205,99){\small $(R_{101}\mld Q_{101})\mlc \overline{R}_{101}$}

\put(112,99){\small $(R_{110}\mld Q_{110})\mlc \overline{R}_{110}$}

\put(-1,35){\small $ P_{1111}\mlc (\overline{P}_{1111} \mld\overline{Q}_{1111})$}
\put(193,35){\small $(R_{1111}\mld Q_{1111})\mlc \overline{R}_{1111}$}
\put(193,67){\small $(R_{1110}\mld Q_{1110})\mlc \overline{R}_{1110}$}

\put(-1,67){\small $ P_{1110}\mlc (\overline{P}_{1110} \mld\overline{Q}_{1110})$}
\put(296,67){\small $(R_{001}\mld Q_{001})\mlc \overline{R}_{001}$}
\put(389,67){\small ${\cal R}_{01}$}

\put(46,87){\line(0,1){10}}
\put(46,87){\line(-3,-1){37}}

\put(363,55){\line(1,0){30}}
\put(363,55){\line(0,1){10}}
\put(393,55){\line(0,1){10}}

\put(268,55){\line(1,0){36}}
\put(268,55){\line(0,1){10}}
\put(304,55){\line(0,1){10}}

\put(72,55){\line(1,0){161}}
\put(72,55){\line(0,1){10}}
\put(233,55){\line(0,1){10}}

\put(198,10){\bf Figure 6}

\end{picture}
\end{center}

Do you see a pattern here? As an exercise, try to trace three more steps, namely, in the scenario where the environment replicates ${\cal R}_{00}$, then ${\cal R}_{10}$, and then ${\cal R}_{11}$. Once you are done, you have understood the strategy and there is no need to read our further --- general --- description of it. 

In general terms, the work of $\cal K$ is divided into {\bf stages}, with each stage represented by a diagram in the style of Figures 3-6. Figure 3 shows  the first stage, which includes initialization as described earlier and subsequent maintainance of synchronization between three pairs of subgames. 

Now let us consider stage $\# n$ for an arbitrary $n\geq 1$. The corresponding diagram will look like the one shown in Figure 7. Here we provide additional explanations for that diagram:
\begin{itemize}
  \item $w_1,\ldots,w_n$ stand for the bit strings representing the ``currently existing'' threads of the $!$-component; there are exactly $n$ such threads. 
  \item  $z$ is $1^n$ (the string of $n$ ``$1$''s). Both of the $?$-components have thread $z$, reserved for the purpose of keeping a ``fresh copy'' of the corresponding game (game $P\mlc (\overline{P}\mld \overline{Q})$ in the left $?$-component and game $(R\mld Q)\mlc  \overline{R}$ in the right $?$-component).
  \item $u_1,\ldots,u_n$, in addition to $z$, stand for the bit strings representing the ``currently existing'' threads of the left $?$-component; there are exactly $n+1$ such threads. Here  each $u_i$ is $1^{i-1}0$ (the string of $i-1$ ``$1$''s followed by a ``$0$'').   
  \item  $k_1,\ldots,k_n$ are positive integers and, for each $j\in\{1,\ldots,n\}$ and $e\in\{1,\ldots,k_j\}$, $v_{j}^{e}$ is a bit string that represents some thread of the right $?$-component. Together with $z$, such $v_{j}^{e}$s are all of the (pairwise distinct) threads of that component.
\end{itemize}

\begin{center} \begin{picture}(432,305)
\put(32,285){\small $\overline{\cal P}$}
\put(36,272){\line(0,1){10}}
\put(36,272){\line(-3,-1){31}}

\put(0,253){\small $ P_{u_{1}}\mlc \hspace{1pt}(\overline{P}_{u_{1}} \mld\hspace{1pt}\overline{Q}_{u_{1}})$}
\put(100,253){\small $(R_{v_{1}^{1}}\mld Q_{v_{1}^{1}})\hspace{1pt}\mlc\hspace{1pt} \overline{R}_{v_{1}^{1}}$}
\put(419,253){\small ${\cal R}_{w_{1}}$}
\put(346,242){\small $\times$}
\put(225,242){\small $\times$}
\put(194,253){\small $(R_{v_{1}^{2}}\mld \hspace{1pt}Q_{v_{1}^{2}})\hspace{1pt}\mlc\hspace{1pt} \overline{R}_{v_{1}^{2}}$}
\put(310,253){\small $(R_{v_{1}^{k_{1}}}\mld\hspace{1pt} Q_{v_{1}^{k_{1}}})\hspace{2pt}\mlc\hspace{2pt} \overline{R}_{v_{1}^{k_{1}}}$}
\put(286,253){\large $\ldots$}

\put(36,240){\line(0,1){10}}
\put(36,240){\line(-3,-1){15}}

\put(65,240){\line(1,0){69}}
\put(65,240){\line(0,1){10}}
\put(134,240){\line(0,1){10}}
\put(166,240){\line(1,0){36}}
\put(166,240){\line(0,1){10}}
\put(202,240){\line(0,1){10}}
\put(260,240){\line(1,0){10}}
\put(260,240){\line(0,1){10}}
\put(318,240){\line(0,1){10}}
\put(318,240){\line(-1,0){10}}
\put(387,240){\line(1,0){36}}
\put(387,240){\line(0,1){10}}
\put(423,240){\line(0,1){10}}

\put(218,225){$\bullet$} 
\put(218,219){$\bullet$} 
\put(218,213){$\bullet$}

\put(0,193){\small $ P_{u_{i-1}}\hspace{-3pt}\mlc\hspace{-2pt} (\overline{P}_{u_{i-1}}\hspace{-4pt} \mld\hspace{-2pt}\overline{Q}_{u_{i-1}}\hspace{-1pt})$}
\put(100,193){\small $(R_{v_{i-1}^{1}}\hspace{-5pt}\mld\hspace{-2pt} Q_{v_{i-1}^{1}}\hspace{-1pt})\hspace{-1pt}\mlc\hspace{-1pt} \overline{R}_{v_{i-1}^{1}}$}
\put(419,193){\small ${\cal R}_{w_{i-1}}$}
\put(345,182){\small $\times$}
\put(224,182){\small $\times$}
\put(194,193){\small $(R_{v_{i-1}^{2}}\hspace{-5pt}\mld\hspace{-2pt} Q_{v_{i-1}^{2}}\hspace{-1pt})\hspace{-1pt}\mlc\hspace{-1pt} \overline{R}_{v_{i-1}^{2}}$}
\put(310,193){\small $(R_{v_{i-1}^{k_{i-1}}}\hspace{-4pt}\mld\hspace{-2pt} Q_{v_{i-1}^{k_{i-1}}})\hspace{-1pt}\mlc\hspace{-1pt} \overline{R}_{v_{i-1}^{k_{i-1}}}$}
\put(286,193){\large $\ldots$} 

\put(6,202){\line(3,1){15}}
\put(36,180){\line(0,1){10}}
\put(36,180){\line(-3,-1){31}}
\put(65,180){\line(1,0){69}}
\put(65,180){\line(0,1){10}}
\put(134,180){\line(0,1){10}}
\put(166,180){\line(1,0){36}}
\put(166,180){\line(0,1){10}}
\put(202,180){\line(0,1){10}}
\put(260,180){\line(1,0){10}}
\put(260,180){\line(0,1){10}}
\put(318,180){\line(0,1){10}}
\put(318,180){\line(-1,0){10}}
\put(387,180){\line(1,0){36}}
\put(387,180){\line(0,1){10}}
\put(423,180){\line(0,1){10}}

\put(0,161){\small $ P_{u_{i }}\hspace{1pt}\mlc\hspace{1pt} (\hspace{1pt}\overline{P}_{u_{i }} \mld\hspace{1pt}\overline{Q}_{u_{i }})$}
\put(100,161){\small $(R_{v_{i }^{1}}\mld \hspace{1pt}Q_{v_{i }^{1}})\hspace{1pt}\mlc\hspace{1pt} \overline{R}_{v_{i }^{1}}$}
\put(419,161){\small ${\cal R}_{w_{i }}$}
\put(345,150){\small $\times$}
\put(224,150){\small $\times$}
\put(194,161){\small $(R_{v_{i }^{2}}\mld \hspace{1pt}Q_{v_{i }^{2}})\hspace{1pt}\mlc\hspace{1pt} \overline{R}_{v_{i }^{2}}$}
\put(310,161){\small $(R_{v_{i }^{k_{i }}}\mld \hspace{2pt}Q_{v_{i }^{k_{i }}})\hspace{2pt}\mlc \hspace{3pt}\overline{R}_{v_{i }^{k_{i }}}$}
\put(286,161){\large $\ldots$} 

\put(36,148){\line(0,1){10}}
\put(36,148){\line(-3,-1){31}}
\put(65,148){\line(1,0){69}}
\put(65,148){\line(0,1){10}}
\put(134,148){\line(0,1){10}}
\put(166,148){\line(1,0){36}}
\put(166,148){\line(0,1){10}}
\put(202,148){\line(0,1){10}}
\put(260,148){\line(1,0){10}}
\put(260,148){\line(0,1){10}}
\put(318,148){\line(0,1){10}}
\put(318,148){\line(-1,0){10}}
\put(387,148){\line(1,0){36}}
\put(387,148){\line(0,1){10}}
\put(423,148){\line(0,1){10}}

\put(0,129){\small $ P_{u_{i+1}}\hspace{-3pt}\mlc\hspace{-2pt} (\overline{P}_{u_{i+1}}\hspace{-4pt} \mld\hspace{-2pt}\overline{Q}_{u_{i+1}}\hspace{-1pt})$}
\put(100,129){\small $(R_{v_{i+1}^{1}}\hspace{-5pt}\mld\hspace{-2pt} Q_{v_{i+1}^{1}}\hspace{-1pt})\hspace{-1pt}\mlc\hspace{-1pt} \overline{R}_{v_{i+1}^{1}}$}
\put(419,129){\small ${\cal R}_{w_{i+1}}$}
\put(345,118){\small $\times$}
\put(224,118){\small $\times$}
\put(194,129){\small $(R_{v_{i+1}^{2}}\hspace{-5pt}\mld\hspace{-2pt} Q_{v_{i+1}^{2}}\hspace{-1pt})\hspace{-1pt}\mlc\hspace{-1pt} \overline{R}_{v_{i+1}^{2}}$}
\put(310,129){\small $(R_{v_{i+1}^{k_{i+1}}}\hspace{-4pt}\mld\hspace{-2pt} Q_{v_{i+1}^{k_{i+1}}})\hspace{-1pt}\mlc\hspace{-1pt} \overline{R}_{v_{i+1}^{k_{i+1}}}$}
\put(286,129){\large $\ldots$} 

\put(36,116){\line(0,1){10}}
\put(36,116){\line(-3,-1){15}}
\put(65,116){\line(1,0){69}}
\put(65,116){\line(0,1){10}}
\put(134,116){\line(0,1){10}}
\put(166,116){\line(1,0){36}}
\put(166,116){\line(0,1){10}}
\put(202,116){\line(0,1){10}}
\put(260,116){\line(1,0){10}}
\put(260,116){\line(0,1){10}}
\put(318,116){\line(0,1){10}}
\put(318,116){\line(-1,0){10}}
\put(387,116){\line(1,0){36}}
\put(387,116){\line(0,1){10}}
\put(423,116){\line(0,1){10}}

\put(218,101){$\bullet$} 
\put(218,95){$\bullet$} 
\put(218,89){$\bullet$} 

\put(0,69){\small $ P_{u_{n}}\mlc \hspace{1pt}(\overline{P}_{u_{n}} \mld\hspace{1pt}\overline{Q}_{u_{n}})$}
\put(100,69){\small $(R_{v_{n}^{1}}\mld Q_{v_{n}^{1}})\hspace{1pt}\mlc\hspace{1pt} \overline{R}_{v_{n}^{1}}$}
\put(419,69){\small ${\cal R}_{w_{n}}$}
\put(346,58){\small $\times$}
\put(225,58){\small $\times$}
\put(194,69){\small $(R_{v_{n}^{2}}\mld \hspace{1pt}Q_{v_{n}^{2}})\hspace{1pt}\mlc\hspace{1pt} \overline{R}_{v_{n}^{2}}$}
\put(310,69){\small $(R_{v_{n}^{k_{n}}}\mld\hspace{1pt} Q_{v_{n}^{k_{n}}})\hspace{2pt}\mlc\hspace{2pt} \overline{R}_{v_{n}^{k_{n}}}$}
\put(286,69){\large $\ldots$}

\put(65,56){\line(1,0){69}}
\put(65,56){\line(0,1){10}}
\put(134,56){\line(0,1){10}}
\put(166,56){\line(1,0){36}}
\put(166,56){\line(0,1){10}}
\put(202,56){\line(0,1){10}}
\put(260,56){\line(1,0){10}}
\put(260,56){\line(0,1){10}}
\put(318,56){\line(0,1){10}}
\put(318,56){\line(-1,0){10}}
\put(387,56){\line(1,0){36}}
\put(387,56){\line(0,1){10}}
\put(423,56){\line(0,1){10}}
\put(6,78){\line(3,1){15}}

\put(0,35){\small $ P_{z}\mlc (\overline{P}_{z} \mld\overline{Q}_{z})$}
\put(100,35){\small $(R_{z}\mld Q_{z})\mlc \overline{R}_{z}$}

\put(198,10){\bf Figure 7}

\end{picture}
\end{center}

The work of $\cal K$ during stage $\#n$ consists in performing the synchronization routine represented by the arcs of Figure 7. Let us observe right now (through a routine analysis left to the reader) that, if stage $\#n$ lasts forever, $\cal K$ wins. As an aside, in this case $\cal K$ wins even if $!$ means $\st$ rather than $\bst$.

Stage $\#n$ will end if and when the environment splits  one of the threads $w_1,\ldots,w_n$ of the $!$-component. Let us assume thread $w_i$ is split, which now becomes two threads: $w_i0$ and $w_i1$.    
This triggers a transition to stage $\#(n+1)$. A diagram for that stage is  shown in Figure 8. 
Note that we have placed the newly emerged thread $w_i1$ of the $!$-component at the bottom of the list of (non-$z$) threads of that component, while leaving thread $w_i0$ where thread $w_i$ was previously found. In response to the above replicative move by the environment, $\cal K$ makes a series of replications. Namely, it replicates:
\begin{itemize}
  \item  The threads $v_{i}^{1},\ldots,v_{i}^{k_i}$ of the right $?$-component --- the ones that were (in Figure 7) found in the same row as ${\cal R}_{w_i}$. This results in  two series of new threads that replace the old ones:    
  $v_{i}^{1}0,\ldots,v_{i}^{k_i}0$ and $v_{i}^{1}1,\ldots,v_{i}^{k_i}1$. 
   \item The thread $z$ of the right $?$-component, which now turns into $z0$ (i.e. $1^n0$) and $z1$ (i.e. $1^{n+1}$). 
   \item  The thread $z$ of the left $?$-component, which now turns into $z0$ (i.e. $1^n0$) and $z1$ (i.e. $1^{n+1}$). 
\end{itemize}
Where these newly born threads/copies are placed, and how the synchronization arrangements are set up or redefined for them (while preserving all other old matchings) can be seen from Figure 8. 

\begin{center} \begin{picture}(440,337)

\put(32,317){\small $\overline{\cal P}$}
\put(36,304){\line(0,1){10}}
\put(36,304){\line(-3,-1){31}}

\put(0,285){\small $ P_{u_{1}}\mlc \hspace{1pt}(\overline{P}_{u_{1}} \mld\hspace{1pt}\overline{Q}_{u_{1}})$}
\put(100,285){\small $(R_{v_{1}^{1}}\mld Q_{v_{1}^{1}})\hspace{1pt}\mlc\hspace{1pt} \overline{R}_{v_{1}^{1}}$}
\put(419,285){\small ${\cal R}_{w_{1}}$}
\put(346,274){\small $\times$}
\put(225,274){\small $\times$}
\put(194,285){\small $(R_{v_{1}^{2}}\mld \hspace{1pt}Q_{v_{1}^{2}})\hspace{1pt}\mlc\hspace{1pt} \overline{R}_{v_{1}^{2}}$}
\put(310,285){\small $(R_{v_{1}^{k_{1}}}\mld\hspace{1pt} Q_{v_{1}^{k_{1}}})\hspace{2pt}\mlc\hspace{2pt} \overline{R}_{v_{1}^{k_{1}}}$}
\put(286,285){\large $\ldots$}

\put(36,272){\line(0,1){10}}
\put(36,272){\line(-3,-1){15}}

\put(65,272){\line(1,0){69}}
\put(65,272){\line(0,1){10}}
\put(134,272){\line(0,1){10}}
\put(166,272){\line(1,0){36}}
\put(166,272){\line(0,1){10}}
\put(202,272){\line(0,1){10}}
\put(260,272){\line(1,0){10}}
\put(260,272){\line(0,1){10}}
\put(318,272){\line(0,1){10}}
\put(318,272){\line(-1,0){10}}
\put(387,272){\line(1,0){36}}
\put(387,272){\line(0,1){10}}
\put(423,272){\line(0,1){10}}

\put(219,257){$\bullet$} 
\put(219,251){$\bullet$} 
\put(219,245){$\bullet$}

\put(0,225){\small $ P_{u_{i-1}}\hspace{-3pt}\mlc\hspace{-2pt} (\overline{P}_{u_{i-1}}\hspace{-4pt} \mld\hspace{-2pt}\overline{Q}_{u_{i-1}}\hspace{-1pt})$}
\put(100,225){\small $(R_{v_{i-1}^{1}}\hspace{-5pt}\mld\hspace{-2pt} Q_{v_{i-1}^{1}}\hspace{-1pt})\hspace{-1pt}\mlc\hspace{-1pt} \overline{R}_{v_{i-1}^{1}}$}
\put(419,225){\small ${\cal R}_{w_{i-1}}$}
\put(345,214){\small $\times$}
\put(224,214){\small $\times$}
\put(194,225){\small $(R_{v_{i-1}^{2}}\hspace{-5pt}\mld\hspace{-2pt} Q_{v_{i-1}^{2}}\hspace{-1pt})\hspace{-1pt}\mlc\hspace{-1pt} \overline{R}_{v_{i-1}^{2}}$}
\put(310,225){\small $(R_{v_{i-1}^{k_{i-1}}}\hspace{-4pt}\mld\hspace{-2pt} Q_{v_{i-1}^{k_{i-1}}})\hspace{-1pt}\mlc\hspace{-1pt} \overline{R}_{v_{i-1}^{k_{i-1}}}$}
\put(286,225){\large $\ldots$} 

\put(6,234){\line(3,1){15}}
\put(36,212){\line(0,1){10}}
\put(36,212){\line(-3,-1){31}}
\put(65,212){\line(1,0){69}}
\put(65,212){\line(0,1){10}}
\put(134,212){\line(0,1){10}}
\put(166,212){\line(1,0){36}}
\put(166,212){\line(0,1){10}}
\put(202,212){\line(0,1){10}}
\put(260,212){\line(1,0){10}}
\put(260,212){\line(0,1){10}}
\put(318,212){\line(0,1){10}}
\put(318,212){\line(-1,0){10}}
\put(387,212){\line(1,0){36}}
\put(387,212){\line(0,1){10}}
\put(423,212){\line(0,1){10}}

\put(0,193){\small $ P_{u_{i}}\hspace{1pt}\mlc\hspace{1pt} (\hspace{1pt}\overline{P}_{u_{i }} \mld\hspace{1pt}\overline{Q}_{u_{i }})$}
\put(100,193){\small $(R_{v_{i}^{1}0}\hspace{-1pt}\mld \hspace{-1pt}Q_{v_{i }^{1}0})\hspace{0pt}\mlc\hspace{0pt} \overline{R}_{v_{i }^{1}0}$}
\put(419,193){\small ${\cal R}_{w_{i }0}$}
\put(345,182){\small $\times$}
\put(224,182){\small $\times$}
\put(194,193){\small $(R_{v_{i }^{2}0}\hspace{-1pt}\mld \hspace{-1pt}Q_{v_{i }^{2}0})\hspace{0pt}\mlc\hspace{0pt} \overline{R}_{v_{i }^{2}0}$}
\put(310,193){\small $(R_{v_{i }^{k_{i }}0}\hspace{-1pt}\mld \hspace{0pt}Q_{v_{i }^{k_{i }}0})\hspace{1pt}\mlc \hspace{2pt}\overline{R}_{v_{i }^{k_{i }}0}$}
\put(286,193){\large $\ldots$} 

\put(36,180){\line(0,1){10}}
\put(36,180){\line(-3,-1){31}}
\put(65,180){\line(1,0){69}}
\put(65,180){\line(0,1){10}}
\put(134,180){\line(0,1){10}}
\put(166,180){\line(1,0){36}}
\put(166,180){\line(0,1){10}}
\put(202,180){\line(0,1){10}}
\put(260,180){\line(1,0){10}}
\put(260,180){\line(0,1){10}}
\put(318,180){\line(0,1){10}}
\put(318,180){\line(-1,0){10}}
\put(387,180){\line(1,0){36}}
\put(387,180){\line(0,1){10}}
\put(423,180){\line(0,1){10}}

\put(0,161){\small $ P_{u_{i+1}}\hspace{-3pt}\mlc\hspace{-2pt} (\overline{P}_{u_{i+1}}\hspace{-4pt} \mld\hspace{-2pt}\overline{Q}_{u_{i+1}}\hspace{-1pt})$}
\put(100,161){\small $(R_{v_{i+1}^{1}}\hspace{-5pt}\mld\hspace{-2pt} Q_{v_{i+1}^{1}}\hspace{-1pt})\hspace{-1pt}\mlc\hspace{-1pt} \overline{R}_{v_{i+1}^{1}}$}
\put(419,161){\small ${\cal R}_{w_{i+1}}$}
\put(345,150){\small $\times$}
\put(224,150){\small $\times$}
\put(194,161){\small $(R_{v_{i+1}^{2}}\hspace{-5pt}\mld\hspace{-2pt} Q_{v_{i+1}^{2}}\hspace{-1pt})\hspace{-1pt}\mlc\hspace{-1pt} \overline{R}_{v_{i+1}^{2}}$}
\put(310,161){\small $(R_{v_{i+1}^{k_{i+1}}}\hspace{-4pt}\mld\hspace{-2pt} Q_{v_{i+1}^{k_{i+1}}})\hspace{-1pt}\mlc\hspace{-1pt} \overline{R}_{v_{i+1}^{k_{i+1}}}$}
\put(286,161){\large $\ldots$} 

\put(36,148){\line(0,1){10}}
\put(36,148){\line(-3,-1){15}}
\put(65,148){\line(1,0){69}}
\put(65,148){\line(0,1){10}}
\put(134,148){\line(0,1){10}}
\put(166,148){\line(1,0){36}}
\put(166,148){\line(0,1){10}}
\put(202,148){\line(0,1){10}}
\put(260,148){\line(1,0){10}}
\put(260,148){\line(0,1){10}}
\put(318,148){\line(0,1){10}}
\put(318,148){\line(-1,0){10}}
\put(387,148){\line(1,0){36}}
\put(387,148){\line(0,1){10}}
\put(423,148){\line(0,1){10}}

\put(219,133){$\bullet$} 
\put(219,127){$\bullet$} 
\put(219,121){$\bullet$} 

\put(0,101){\small $ P_{u_{n}}\mlc \hspace{1pt}(\overline{P}_{u_{n}} \mld\hspace{1pt}\overline{Q}_{u_{n}})$}
\put(100,101){\small $(R_{v_{n}^{1}}\mld Q_{v_{n}^{1}})\hspace{1pt}\mlc\hspace{1pt} \overline{R}_{v_{n}^{1}}$}
\put(419,101){\small ${\cal R}_{w_{n}}$}
\put(346,90){\small $\times$}
\put(225,90){\small $\times$}
\put(194,101){\small $(R_{v_{n}^{2}}\mld \hspace{1pt}Q_{v_{n}^{2}})\hspace{1pt}\mlc\hspace{1pt} \overline{R}_{v_{n}^{2}}$}
\put(310,101){\small $(R_{v_{n}^{k_{n}}}\mld\hspace{1pt} Q_{v_{n}^{k_{n}}})\hspace{2pt}\mlc\hspace{2pt} \overline{R}_{v_{n}^{k_{n}}}$}
\put(286,101){\large $\ldots$}

\put(65,88){\line(1,0){69}}
\put(65,88){\line(0,1){10}}
\put(134,88){\line(0,1){10}}
\put(166,88){\line(1,0){36}}
\put(166,88){\line(0,1){10}}
\put(202,88){\line(0,1){10}}
\put(260,88){\line(1,0){10}}
\put(260,88){\line(0,1){10}}
\put(318,88){\line(0,1){10}}
\put(318,88){\line(-1,0){10}}
\put(387,88){\line(1,0){36}}
\put(387,88){\line(0,1){10}}
\put(423,88){\line(0,1){10}}
\put(6,110){\line(3,1){15}}

\put(0,69){\small $ P_{z0}\hspace{1pt}\mlc\hspace{1pt} (\hspace{1pt}\overline{P}_{z0}\hspace{1pt} \mld\hspace{1pt}\overline{Q}_{z0})$}
\put(100,69){\small $(R_{z0}\mld \hspace{1pt}Q_{z0})\hspace{1pt}\mlc\hspace{2pt} \overline{R}_{z0}$}
\put(419,69){\small ${\cal R}_{w_{i }1}$}
\put(345,58){\small $\times$}
\put(224,58){\small $\times$}
\put(194,69){\small $(R_{v_{i }^{1}1}\hspace{-2pt}\mld \hspace{0pt}Q_{v_{i }^{1}1})\hspace{0pt}\mlc\hspace{0pt} \overline{R}_{v_{i }^{1}1}$}
\put(310,69){\small $(R_{v_{i }^{k_{i }}1}\hspace{-1pt}\mld \hspace{0pt}Q_{v_{i }^{k_{i }}1})\hspace{1pt}\mlc \hspace{0pt}\overline{R}_{v_{i }^{k_{i }}1}$}
\put(286,69){\large $\ldots$} 

\put(36,88){\line(0,1){10}}
\put(36,88){\line(-3,-1){31}}
\put(65,56){\line(1,0){69}}
\put(65,56){\line(0,1){10}}
\put(134,56){\line(0,1){10}}
\put(166,56){\line(1,0){36}}
\put(166,56){\line(0,1){10}}
\put(202,56){\line(0,1){10}}
\put(260,56){\line(1,0){10}}
\put(260,56){\line(0,1){10}}
\put(318,56){\line(0,1){10}}
\put(318,56){\line(-1,0){10}}
\put(387,56){\line(1,0){36}}
\put(387,56){\line(0,1){10}}
\put(423,56){\line(0,1){10}}

\put(0,35){\small $ P_{z1}\mlc (\overline{P}_{z1} \mld\overline{Q}_{z1})$}
\put(100,35){\small $(R_{z1}\mld Q_{z1})\mlc \overline{R}_{z1}$}

\put(200,10){\bf Figure 8}

\end{picture}
\end{center}

We have already observed that, if there are only finitely many stages (i.e., the environment only makes finitely many replications in the $!$-component), $\cal K$ is the winner. It remains to understand why $\cal K$ also wins the game  in the cases where there is no last stage. 

Let us use the term ``{\bf line}'' for the rows of a diagram in the style of Figure 7, for the exception of the topmost row consisting of $\overline{\cal P}$ and the bottommost row
consisting of the two ``reserve'' threads. We number the lines of a diagram consecutively from top to bottom. To see that the strategy $\cal K$ is successful, consider any essentially finite (but possibly infinite) bit string $s$ such that the thread $s$ has actually emerged in the $!$-component. Let $v$ be the shortest (possibly empty) initial segment of $s$ containing all ``$1$''s that $s$ contains. So, $s$ looks like $vu$, where the $u$ part entirely consists of (finitely or infinitely many) ``$0$''s. Remember these $v$ and $u$.

Imagine the stage --- let it be stage $\#n$ --- at which the component ${\cal R}_v$ first emerged in the corresponding diagram. At that time, ${\cal R}_v$ will be placed in line $\#n$. Temporarily assuming that $n$ is the last stage that ever emerges in the play, and 
analyzing\footnote{In fact, we have already implicitly undertaken such an analysis earlier when observing that $\cal K$ is successful in the cases of finitely many stages.}  the diagram for stage $\#n$, one can easily see that 
\begin{equation}\label{just}\begin{array}{l}
\mbox{\em if $\cal K$ loses in thread $v$, then it either wins in the recurrence-free component, or in}\\
\mbox{\em one of the threads $x$ of one of the two $?$-components in line $\#m$ for some $1\leq m\leq n$}.
\end{array}
\end{equation}
Now back to the case of infinitely many stages. An analysis of the situation, which the reader should partly undertake on his or her own 
 after internalizing our construction, reveals that,  in the continuously evolving diagram, everything --- including matching arrangements --- in lines $\#1$ through $\#n$ will remain ``the same'', and the only changes that may be occurring in those lines  are that ``$0$''s will  be added to the names of the threads of the right (but not the left) $?$-component and of the thread of the $!$-component.\footnote{If we imagine an infinite diagram corresponding to the entire play, its lines $\#1$ through $\#n$ will remain ``essentially the same'' as in the diagram for stage $\#n$.} The above changes are of no relevance to our earlier argument for (\ref{just}), in which we now should simply replace ``$v$'' by 
``$vu$'' and ``$x$'' by ``$xy$'', where $y$, just like $u$, is a certain (possibly infinite) string entirely consisting of ``$0$''s; as $xy$ is essentially finite, we again find that $\cal K$  wins the overall game. 

Since the above $s$ was an arbitrary essentially finite string, we conclude that, if $\cal K$ loses in the $!$-component, then it wins the overall game. And, of course, it also wins if it does not lose in the $!$-component.

\section{Uncountable branching recurrence does not validate long production}

An {\bf enumeration game} is a game where any natural number, identified with its decimal representation, is a legal move by either player at any time (and there are no other legal moves). This way, either player can be seen to enumerate a set of numbers --- the numbers made by it as moves during the play. The winner in a (legal) play of an enumeration game only depends on the two sets enumerated this way. That is, what matters is only {\em what} moves have been eventually made and by whom, regardless of {\em when} (in what order) and {\em how many times} those moves were made. 

Let us rewrite long production (\ref{n2}) in the following form:
   \begin{equation}\label{ee1}
\overline{P}^1\ \mld \ ?\bigl(P^2\mlc (\overline{P}^3\mld \overline{P}^4)\bigr)\ \mld\  ?\bigl((P^5\mld P^6)\mlc \overline{P}^7\bigr)\ \mld \ !\hspace{1pt}{P}^8.
\end{equation}
Here we have used the superscripts $1$ through $8$ to differentiate between the different occurrences of $P$, and otherwise it is understood that each $P^i$ simply means $P$. Throughout this section, $!$ stands for $\st$, and $?$ for its  dual  $\cost$ ($=\gneg\st\gneg$).

Let us fix an HPM $\cal H$ as an arbitrary strategy of the machine. We want to construct a {\em counterstrategy} $\cal D$ --- in the same sense as in Section \ref{s4} --- such that, when the environment follows it, $\cal H$ loses (\ref{ee1})  with $P$ interpreted as a certain enumeration game. 

As was done earlier, terminologically we identify formulas with games. This does not create any confusion: because $P$ is going to be interpreted as an enumeration game anyway, the legal moves of it --- even if not the winner --- are known even before we actually define that interpretation.   

As in the preceding section, we refer to $\overline{P}^1$ as the {\em recurrence-free component}, refer to $?\bigl(P^2\mlc (\overline{P}^3\mld \overline{P}^4)\bigr)$ as the {\em left $?$-component}, refer to $?\bigl((P^5\mld P^6)\mlc \overline{P}^7\bigr)$ as the {\em right $?$-component} and refer to $!\hspace{1pt}{P}^8$ as the {\em $!$-component}. A {\bf fresh move} means a number that has not yet been chosen in the play by either player as a move in any subgame $P^{i}$ or $\overline{P}^j$. 

The work of the counterstrategy $\cal D$ is very simple. It consists in repeating, over and over infinitely many times, the following  routine:

\begin{description}
  \item[Step (i):] Split (make a replicative move in) each thread of the $!$-component.
  \item[Step (ii):] Make a fresh move\footnote{Several moves are made during this step, and it is understood that the condition of their ``freshness'' implies that they are  different not only from all earlier moves/numbers, but also from each other.} in the recurrence-free component, in all three subgames of each thread of the left $?$-component, in all three subgames of each thread of the right $?$-component, and in each thread of the $!$-component.
\end{description}     

The rest of our discussion is in the context of the run/play generated in the scenario where $\cal H$ plays as $\pp$ and the environment ($\oo$) acts according to strategy $\cal D$. It is obvious that (as along as $P$ is interpreted as an enumeration game) $\cal D$ plays legally. As always, we safely assume that its adversary $\cal H$ never makes illegal moves, either.  

Let $\mathbb{T}^{?}_{l}$ be the set of all threads that eventually emerge in the left $?$-component, $\mathbb{T}^{?}_{r}$ be the set of all threads that eventually emerge in the right $?$-component, and $\mathbb{T}^{!}$ be the set of all threads that eventually emerge in the  $!$-component. As an aside, notice that $\mathbb{T}^{!}$ is exactly the set of all infinite bit strings. The same is not necessarily the case for  $\mathbb{T}^{?}_{l}$ and 
 $\mathbb{T}^{?}_{r}$ though, which may even be finite. 

What we here --- by some abuse of terminology --- call {\bf literals} are the following objects (technically, these ``objects'' are nothing but superscript/subscript pairs, except for one case where we only have a superscript):

\begin{itemize}
  \item $\overline{P}^{1}$;
  \item $P^{2}_{w}$,  $\overline{P}^{3}_{w}$ and $\overline{P}^{4}_{w}$ for each $w\in\mathbb{T}^{?}_{l}$;
  \item $P^{5}_{w}$, $P^{6}_{w}$ and $\overline{P}^{7}_{w}$ for each $w\in\mathbb{T}^{?}_{r}$;
  \item $P^{8}_{w}$  for each $w\in\mathbb{T}^{!}$.
\end{itemize} 

The literals of the form $P^{2}_{w}$, $P^{5}_{w}$, $P^{6}_{w}$,  $P^{8}_{w}$ are {\bf positive}, and all other literals (the ones with an overline) are {\bf negative}.

With each literal $L$ we associate its {\bf content}. The latter is a pair $(L_\pp,L_\oo)$,  where $L_\pp$ is the set of all moves (numbers) made (enumerated) by $\cal H$ in the corresponding (sub)game, and $L_\oo$ is the set of all moves made by $\cal D$ there. Here the meaning of ``the corresponding (sub)game''  must be clear: in the case of $\overline{P}^{1}$, this is the (sub)game played in the recurrence-free component; in the case of 
$P^{2}_{w}$ (resp. $\overline{P}^{3}_{w}$,  $\overline{P}^{4}_{w}$), this is the (sub)game played in the $P^{2}$  (resp. $\overline{P}^{3}$,  $\overline{P}^{4}$) part of thread $w$ of the left $?$-component; in the case of 
$P^{5}_{w}$ (resp. $P^{6}_{w}$,  $\overline{P}^{7}_{w}$), this is the (sub)game played in the $P^{5}$  (resp. $P^{6}$,  $\overline{P}^{7}$) part of thread $w$ of the right $?$-component; and, in the case of $P^{8}_{w}$, this is the (sub)game played in thread $w$ of the  $!$-component.
 
\begin{lemma}\label{july5a}
%\marginpar{july5a}
If $L$ and $M$ are two distinct literals (meaning that either their superscripts, or their subscripts, or both, are non-identical), then  their contents are also distinct.  
\end{lemma}
\begin{proof} If $L$ and $M$ are two distinct literals, the (sub)games they represent happen to be in different parts of the overall game (including the possibility of being in different threads). But $\cal D$ keeps making fresh moves in all existing components.  So, some number (in fact, infinitely many numbers) is in the set enumerated by $\cal D$ in one (sub)game but not in the set enumerated by $\cal D$  in the other (sub)game.
\end{proof}

Let $L$ and $M$ be literals. We say that $L$ {\bf matches} $M$ iff,  where $(L_\pp,L_\oo)$ and  $(M_\pp,M_\oo)$ are the contents of $L$ and $M$, 
 we have $L_\pp=M_\oo$ and $L_\oo=M_\pp$. Next, we 
say that $L$ is a {\bf threadmate} of $M$ iff $L\not=M$ and one of the following conditions is satisfied:
\begin{itemize}
  \item For some $w\in \mathbb{T}^{?}_{l}$, both $L$ and $M$ are among $P^{2}_{w}, \overline{P}^{3}_{w},\overline{P}^{4}_{w}$. 
  \item For some $w\in \mathbb{T}^{?}_{r}$, both $L$ and $M$ are among $P^{5}_{w}, P^{6}_{w},\overline{P}^{7}_{w}$. 
\end{itemize}  

We define a {\bf chain} as a nonempty finite sequence  $L_1,\ldots,L_n$ of literals satisfying the following conditions:
\begin{itemize}
  \item $L_1$ (and only $L_1$) is $\overline{P}^{1}$.
  \item For each odd $i$ with $1\leq i<n$, $L_{i+1}$ matches $L_{i}$.
  \item For each even $i$ with $1\leq i<n$, $L_{i+1}$ is a threadmate of $L_{i}$.
\end{itemize}  

Let 
\[\overline{P}^{i_1}, \tilde{P}^{i_2}_{w_2}, \tilde{P}^{i_3}_{w_3}, \tilde{P}^{i_4}_{w_4},\ldots,\tilde{P}^{i_n}_{w_n}\]
be a chain (here $\tilde{P}$ stands for $P$ with or without an overline;  and, of course, $i_1=1$).
The {\bf type} of such a chain is the sequence $i_1,i_2,i_3,i_4,\ldots,i_n$ of the superscripts of the above-displayed literals.

\begin{lemma}\label{dec21a}
%\marginpar{dec21a}
There are no two non-identical chains that have the same type. 
\end{lemma}
\begin{proof} Rather immediately from  Lemma \ref{july5a}. 
\end{proof}

Let us say that a literal $L$ is {\bf reachable} iff there is a chain $L_1,\ldots,L_n$ with $L_n=L$. 

\begin{lemma}\label{dec21b}
%\marginpar{dec21b}
There are only countably many reachable literals.
\end{lemma}
\begin{proof} The number of all possible types of chains is countable, because every type is a finite sequence of numbers. And, in view of Lemma \ref{dec21a}, all reachable literals can be listed by listing the (unique) types of the corresponding chains. 
\end{proof}

We select an interpretation $^*$ that interprets $P$ as the enumeration game such that, whenever $(S_\pp,S_\oo)$ is the pair of the sets enumerated (while playing $P^*$) by $\pp$ and $\oo$, respectively, $\pp$ is the winner if and only if $(S_\pp,S_\oo)$ is the content of some reachable positive literal or $(S_\oo,S_\pp)$ is the content of some reachable negative literal. 

Implicitly relying on Lemma \ref{july5a}, we now claim that $\cal H$ loses the overall game under this interpretation. 
To see why, first observe that there is a thread $u$ (in fact, uncountably many such threads) in the $!$-component such that $P^{8}_{u}$ is not reachable. This is so because there are uncountably many threads in the $!$-component, of which, however, according to Lemma \ref{dec21b}, only countably many are reachable. By our choice of interpretation, $\cal H$ loses the game $P$ in thread $u$, meaning that it loses the entire $!$-component. Next, $\cal H$ loses the recurrence-free component because $\overline{P}^1$ is reachable yet negative. Next, $\cal H$ can be seen to lose in every thread $w$ of the left $?$-component. Namely, if one of (which is the same as to say that {\em all of}) the three literals $P^{2}_{w},\overline{P}^{3}_{w},\overline{P}^{4}_{w}$ is reachable, then $\cal H$ loses because it loses in the $\overline{P}^{3}\mld\overline{P}^{4}$ part of the thread; otherwise $\cal H$ loses because it loses in the $P^{2}$ part. Finally, the threads of the right $?$-component can be handled in a similar way.

\end{document}